\RequirePackage{fix-cm}
\documentclass[twocolumn]{svjour3}          
\smartqed  
\usepackage{graphicx}
\usepackage{amssymb,amsmath}
\usepackage{enumerate}
\usepackage[para,online,flushleft]{threeparttable}
\usepackage{multirow}
\usepackage{subfig}
\usepackage[linesnumberedhidden,ruled,vlined]{algorithm2e}
\usepackage{algpseudocode}
\usepackage{url}
\usepackage{hyperref}
\usepackage{color}

\let\oldnl\nl
\newcommand{\nonl}{\renewcommand{\nl}{\let\nl\oldnl}}

\journalname{EURASIP Journal on Advances in Signal Processing}

\begin{document}

\title{Computable performance guarantees for compressed sensing matrices
}


\author{Myung~Cho \and
        Kumar Vijay~Mishra         \and
        Weiyu~Xu 
}


\institute{ \small{Myung~Cho, Kumar Vijay~Mishra, and Weiyu~Xu \at
              Dept. of ECE, University of Iowa, Iowa City, IA, 52242 \\
              \email{(myung-cho, kumarvijay-mishra, weiyu-xu)@uiowa.edu.}}
}

\date{Received: date / Accepted: date}

\maketitle

\begin{abstract}
The null space condition for $\ell_1$ minimization in compressed sensing is a necessary and sufficient condition on the sensing matrices under which a sparse signal can be uniquely recovered from the observation data via $\ell_1$ minimization. However, verifying the null space condition is known to be computationally challenging. Most of the existing methods can provide only upper and lower bounds on the proportion parameter that characterizes the null space condition. In this paper, we propose new polynomial-time algorithms to establish upper bounds of the proportion parameter. We leverage on these techniques to find upper bounds and further develop a new procedure - tree search algorithm - that is able to precisely and quickly verify the null space condition. Numerical experiments show that the execution speed and accuracy of the results obtained from our methods far exceed those of the previous methods which rely on Linear Programming (LP) relaxation and Semidefinite Programming (SDP).
\end{abstract}

\section{Introduction}
\label{intro}
Compressed sensing is an efficient signal processing technique to recover a sparse signal from fewer samples than required by the Nyquist-Shannon theorem, reducing time and energy spent in sampling operation. These advantages make compressed sensing attractive in various signal processing areas \cite{eldar2012compressed}.

In compressed sensing, we are interested in recovering the sparsest vector $x \in \mathbb{R}^n$ that satisfies the underdetermined equation $y=Ax$. Here $\mathbb{R}$ is the set of real numbers, $A \in \mathbb{R}^{m \times n}, \; m < n$, is a sensing matrix, and $y \in \mathbb{R}^m$ is the observation or measurement data. This is posed as an $\ell_0$ minimization problem:
\setlength{\mathindent}{15pt}
\par\noindent
\small
\begin{align}
\label{L0_min}
          & \text{minimize} \;\; \| x \|_{0} \nonumber\\
          & \text{subject to} \;\; y = Ax,
\end{align}
\normalsize
where $\|x\|_0$ is the number of non-zero elements in vector $x$. The $\ell_0$ minimization is an NP-hard problem. Therefore, we often relax (\ref{L0_min}) to its closest convex approximation - the $\ell_1$ minimization problem:
\par\noindent
\small
\begin{align}
\label{L1_min}
         & \text{minimize} \;\; \| x \|_{1}  \nonumber\\
         & \text{subject to} \;\; y = Ax.
\end{align}
\normalsize
It has been shown that the optimal solution of $\ell_0$ minimization can be obtained by solving $\ell_1$ minimization under certain conditions (e.g. Restricted Isometry Property or RIP) \cite{candes2005decoding,d2011testing,candes2006robust,candes2006stable,donoho2005neighborly}. For random sensing matrices, these conditions hold with high probability. We note that RIP is a sufficient condition for sparse recovery \cite{Tillmann2014Computational}.

A necessary and sufficient condition under which a $k$-sparse signal $x$, ($k \ll n$) can be uniquely obtained via $\ell_1$ minimization is Null Space Condition (NSC) \cite{juditsky2011verifiable,cohen2009compressed,d2011testing}. A matrix $A$ satisfies NSC for a positive integer $k$ if
\par\noindent
\small
\begin{align}
    ||z_K ||_1 < || z_{\overline K} ||_1
\end{align}
\normalsize
holds true for all $z \in \{z:\; Az=0, z\neq 0\}$ and for all subsets $K \subseteq \{1,2,...,n\}$ with $|K|\leq k$. Here $K$ is an index set, $|K|$ is the cardinality of $K$, $z_K$ is the part of the vector $z$ over the index set $K$, and $\overline{K}$ is the complement of $K$. NSC is related to the proportion parameter $\alpha_k$ defined as
\par\noindent
\small
\begin{align}
\label{eq:alphak}
    \alpha_k \triangleq \underset{\{z:\; Az=0,\;z\neq 0\}}{\text{maximize}} \underset{{\{K:\; |K|\leq k\}}}{\text{maximize}} \;\;\frac{\|z_K \|_{1}}{\|z\|_{1}}.
\end{align}
\normalsize
The $\alpha_k$ is the optimal value of the following optimization problem:
\par\noindent
\small
\begin{align}
\label{opt_prob}
          &\underset{z,\{K:\;|K|\leq k\}}{\text{maximize}}\;\; \| z_K \|_{1}             \nonumber\\
          &  \;\;\;        \text{subject to} \;\; \| z \|_{1} \leq 1,\; Az = 0,
\end{align}
\normalsize
where $K$ is a subset of $\{1,2, \dots, n\}$ with cardinality at most $k$. The matrix $A$ satisfies NSC for a positive integer $k$ if and only if $\alpha_k<\frac{1}{2}$. Equivalently, NSC can be verified by computing or estimating $\alpha_k$.
The role of $\alpha_k$ is also important in recovery of an approximately sparse signal $x$ via $\ell_1$ minimization where a smaller $\alpha_k$ implies more robustness \cite{juditsky2011verifiable,cohen2009compressed,xu2011precise}.

We are interested in computing $\alpha_k$ and, especially, finding the maximum $k$ for which $\alpha_k < \frac{1}{2}$. However, computing $\alpha_k$ to verify NSC is extremely expensive and was reported in \cite{Tillmann2014Computational} to be NP-hard. Due to the challenges in computing $\alpha_k$, verifying NSC explicitly for deterministic sensing matrices remains a relatively unexamined research area. In \cite{lee2008computing,tang2011verifiable,d2011testing,juditsky2011verifiable}, convex relaxations were used to establish upper or lower bounds of $\alpha_k$ (or other parameters related to $\alpha_k$) instead of computing the exact $\alpha_k$. While \cite{d2011testing,lee2008computing} proposed semidefinite programming based methods, \cite{juditsky2011verifiable,tang2011verifiable} suggested linear programming relaxations to obtain the upper and lower bounds of $\alpha_k$. For both methods, computable performance guarantees on sparse signal recovery were reported via bounding $\alpha_k$. However, these bounds of $\alpha_k$ could only verify NSC with $k=O(\sqrt{n})$, even though theoretically $k$ can grow linearly with $n$.

Our work drastically departs from these prior methods \cite{lee2008computing,tang2011verifiable,d2011testing,juditsky2011verifiable} that provide only the upper and lower bounds. In our solution, we propose the \textit{pick-$l$-element} algorithms ($1 \leq l < k$), which compute upper bounds of $\alpha_k$ in polynomial time. Subsequently, we leverage on these algorithms to develop the \textit{Tree Search Algorithm (TSA)} - a new method to compute an exact $\alpha_k$ by significantly reducing computational complexity of an exhaustive search method. This algorithm offers a way to control a smooth tradeoff between complexity and accuracy of the computations. In the conference precursor to this paper, we had introduced \textit{Sandwiching Algorithm (SWA)} \cite{cho2013new}, which employs a branch-and-bound method. Although SWA can also be used to calculate the exact $\alpha_k$, it has a disadvantage of greater memory usage than TSA. On the other hand, TSA provides memory and performance benefits for high-dimensional matrices (e.g., up to size $\sim 6000 \times 6000$).

It is noteworthy that our methods are different from RIP or the neighborly polytope framework for analyzing the sparse recovery capability of random sensing matrices. For example, prior works such as  \cite{donoho2005neighborly,donoho2010precise} employ the neighborly polytope to predict theoretical lower bounds on recoverable sparsity $k$ for a randomly chosen Gaussian matrix. However, our methods do not resort to a probabilistic analysis and are applicable for any given deterministic sensing matrix. Also, our algorithms have the strength of providing better bounds than existing methods \cite{lee2008computing,tang2011verifiable,d2011testing,juditsky2011verifiable} for a wide range of matrix sizes.
\subsection{Main contributions}
We summarize our main contributions as follows:
\begin{enumerate}[(i)]
\item \textbf{Faster algorithms for high dimensions.} We designed the  pick-$l$ algorithm (and its optimized version), where $l$ is a chosen integer, to provide upper bounds on $\alpha_k$. We are able to show that when $l$ increases, the optimized pick-$l$ algorithm provides tighter upper bound on $\alpha_{k}$. Numerical experiments show that, even with $l=2$ or $3$, the pick-$l$ algorithm already provides better bound on $\alpha_k$ than the previous algorithms based on the LP \cite{juditsky2011verifiable} and SDP \cite{d2011testing}. For large sensing matrices,  the pick-$1$-element algorithm can be significantly faster than the LP and SDP methods.

\item \textbf{Novel formulations using branch-and-bound.} Based on the pick-$l$ algorithm, we propose a branch-and-bound tree search approach to compute tighter bounds or even the exact value of $\alpha_{k}$. To the best of our knowledge, this tree search algorithm is the first branch-and-bound algorithm to verify NSC for $\ell_1$ minimization.  This branch-and-bound approach heavily depends on the pick-$l$ algorithm developed in this paper.  For example, the LP \cite{juditsky2011verifiable} and  SDP \cite{d2011testing} methods  cannot be directly adapted to provide an efficient branch and bound approach,  due to their lack of subset-specific upper bounds on $\alpha_{k}$. In numerical experiments, we demonstrated that the tree search algorithm reduced the execution time to precisely calculate $\alpha_k$ by around 40-8000 times, compared to the exhaustive search method. 

\item \textbf{Simultaneous upper and lower bounds.} The branch-and-bound tree search algorithm simultaneously maintains upper and lower bounds of $\alpha_ k$ during the run-time. This approach has two benefits. Firstly, if one is interested in merely certifying the NSC for a positive $k$ rather than obtaining the exact $\alpha_k$, then one can terminate the TSA early to shorten the running time. This can be done as soon as the global upper (lower) bound drops below (exceeds) 1/2 and, therefore, concluding that the NSC for the positive $k$ is satisfied (not satisfied). Secondly, consider the case when TSA is terminated early due to, say, constraints on running time. Then, the process still yields meaningful bounds on $\alpha_k$ via the record of continuously maintained upper and lower bounds.

\item \textbf{New results on recoverable sparsity.} For a certain $l<k$,  we can compute $\alpha_{{l}} $ or its upper bound by using the branch-and-bound tree search algorithm (for example, based on the pick-$1$-element algorithm). We introduce a novel result (Lemma \ref{kmax_bound}), which can use $\alpha_{l}$ to lower bound the recoverable sparsity $k$. This approach of lower bounding the recoverable sparsity $k$ is useful when $l$ is too large to perform the pick-$l$ algorithm directly (which requires $\binom{n}{l}$ enumerations).
\end{enumerate}

\subsection{Notations and preliminaries}
\label{notation}
We denote the sets of real numbers, and positive integers as $\mathbb{R}$ and $\mathbb{Z}^{+}$ respectively. We reserve uppercase letters $K$ and $L$ for index sets, and lowercase letters $k, l \in \mathbb{Z}^{+}$ for their respective cardinalities. We also use $|\cdot|$ to denote the cardinality of a set. We assume $k > l \geq 1$ throughout the paper. For vectors or scalars, we use lowercase letters, e.g., $x, k, l$. For a vector $x \in \mathbb{R}^{n}$, we use $x_i$ for its $i$-th element. If we use an index set as a subscript of a vector, it represents the partial vector over the index set. For example, when $x \in \mathbb{R}^n$ and $K=\{1,2\}$, $x_K$ represents $[x_1, x_2]^T$. We reserve uppercase $A$ for a sensing matrix whose dimension is $m \times n$. Since the number of columns of a sensing matrix $A$ is $n$, the full index set we consider is $\{1,2,...,n\}$. In addition, we represent $\binom{n}{l}$ numbers of subsets as $L_i$, $i=1,...,\binom{n}{l}$, where $L_i \subset \{1,2,...,n\}$, $|L_i| = l$. We use the superscript * to represent an optimal solution of an optimization problem. For instance, $z^{*}$ and $K^{*}$ are the optimal solution of (\ref{opt_prob}). Since we need to represent an optimal solution for each index set $L_i$, we use the superscript $i*$ to represent an optimal solution for an index set $L_i$, e.g., $z^{i*}$. The maximum value of $k$ such that both $\alpha_{k} < \frac{1}{2}$ and $\alpha_{k+1} \geq \frac{1}{2}$ hold true is denoted by the \textit{maximum recoverable sparsity} $k_{max}$.

\section{Pick-$l$-element Algorithm}
\label{Sec1}
Consider a sensing matrix with $n$ columns. Then, there are $\binom{n}{k}$ subsets $K$ each of cardinality $k$. When $n$ and $k$ are large, exhaustive search over these subsets to compute $\alpha_k$ is extremely expensive. For example, when $n=100$ and $k=10$, it takes a search over 1.7310e+13 subsets to compute $\alpha_k$ - a combinatorial task that is beyond the technological reach of common desktop computers. Our goal is to devise algorithms that can rapidly yield an exact value of $\alpha_k$. As an initial step, we develop a method to compute an upper bound of $\alpha_k$ in polynomial time, which is called the \textit{pick-$l$-element algorithm} (or simply, pick-$l$ algorithm), where $l$ is a chosen integer such that $1 \leq l < k$.

Let us define the proportion parameter for a given index set $L$ such that $|L|=l$, denoted by $\alpha_{l,L}$, as
\par\noindent
\small
\begin{align}
\label{eq:alphalL}
    \alpha_{l,L} \triangleq \underset{\{z:\; Az=0,\;z\neq 0\}}{\text{maximize}} \frac{\|z_L \|_{1}}{\|z\|_{1}}.
\end{align}
\normalsize
(\ref{eq:alphalL}) is the partial optimization problem of (\ref{eq:alphak}) only considering the vector $z$ in the null space of $A$ for a fixed index set $L$. We can obtain $\alpha_{l,L}$ by solving the following optimization problem:
\par\noindent
\small
\begin{align}
\label{opt_prob2}
          & \underset{z}{\text{maximize}} \;\; \| z_L \|_{1} \nonumber\\
          & \text{subject to} \;\; \| z \|_{1} \leq 1, \; Az = 0.
\end{align}
\normalsize
Since (\ref{opt_prob2}) is maximizing a convex function for a given subset $L$, we cast (\ref{opt_prob2}) as $2^{l}$ linear programming problems by considering all the possible sign patterns of every element of $z_L$ (e.g., if $l=2$ and $L=\{1,2\}$, then, $||z_L||_1 = |z_1|+|z_2|$ can correspond to $2^l = 4$ possibilities: $z_1+z_2$, $z_1-z_2$, $-z_1+z_2$, and $-z_1-z_2$). $\alpha_{l,L}$ is equal to the maximum among the $2^{l}$ objective values.

The pick-$l$ algorithm uses $\alpha_{l,L}$'s obtained from different index sets to compute an upper bound of $\alpha_{k}$.  Algorithm \ref{pick_l_algo_pseudo} shows the steps of the pick-$l$ algorithm in detail. The following Lemmata show that the pick-$l$ algorithm provides an upper bound of $\alpha_{k}$. Firstly, we provide Lemma \ref{lemma:cheapupper} to derive the upper bound of the proportion parameter for a fixed index set $K$, and then, we show that the pick-$l$ algorithm yields an upper bound of $\alpha_k$ in Lemma \ref{lemma_pickl_upper_bound}.

\begin{lemma}[Cheap Upper Bound (CUB) for a given subset $K$]
\label{lemma:cheapupper}
 Given a subset $K$, we have
\par\noindent
\small
\begin{align}
\label{eq:cheapbound}
        CUB(\alpha_{k,K}) \triangleq \frac{1}{{\binom{k-1}{l-1}}} \sum_{\{L_i \subseteq K,\; |L_i|=l\}} \alpha_{l,L_i}
        \geq \alpha_{k,K}.
\end{align}
\normalsize
\end{lemma}
\begin{proof}
Suppose that when $z = z^{i*}$ and $z = z^{*}$, we achieve the optimal value of (\ref{eq:alphalL}) for given index sets $L_i$ and $K$ respectively, i.e., $\alpha_{l,L_i} = \frac{ \| z^{i*}_{ L_i  } \|_{1} }{ \| z^{i*} \|_{1} }$ and $\alpha_{k,K} = \frac{\|z^{*}_{K}\|_{1}}{\|z^{*}\|_1}$. Since each element of $K$ appears ${\binom{k-1}{l-1}}$ times in $\{L_i \subseteq K,\; |L_i|=l\}$, we obtain the following inequality:
\par\noindent
\small
\begin{align*}
       \alpha_{k,K} & = \frac{ \| z^{*}_{K} \|_{1} }{ \| z^{*} \|_{1} }  =
       \frac{1}{{\binom{k-1}{l-1}}}  \sum_{ \{L_i \subseteq K,\; |L_i|=l\} } \frac{ \| z^{*}_{ L_i } \|_{1} }{ \| z^{*} \|_{1}} \\
                    & \leq \frac{1}{{\binom{k-1}{l-1}}} \sum_{ \{L_i \subseteq K,\; |L_i|=l\} } \frac{ \| z^{i*}_{ L_i } \|_{1} }{ \| z^{i*} \|_{1} } = CUB(\alpha_{k,K}).
\end{align*}
\normalsize
The inequality is from the optimal value of (\ref{eq:alphalL}) for each index set $L_i$.
\qed
\end{proof}

\begin{lemma}
\label{lemma_pickl_upper_bound}
The pick-$l$ algorithm provides an upper bound of $\alpha_k$, namely
\par\noindent
\small
\begin{align}
\label{pick_l}
       \alpha_k
        \leq
        \frac{1}{{\binom{k-1}{l-1}}}  \sum_{ i=1 }^{\binom{k}{l}} \alpha_{l,L_i},
\end{align}
\vspace{-3mm}
\begin{align}
\label{eq:sorting}
        \text{where}\;\; \alpha_{l,L_1} \geq \alpha_{l,L_2} \geq \cdots \geq \alpha_{l,L_i} \geq \cdots \geq \alpha_{l,L_{\binom{n}{l}}}.
\end{align}
\normalsize
\end{lemma}
\begin{proof}
Without loss of generality, we assume that when $z=z^{i*}$, $i=1,2,...,\binom{n}{l}$, $\alpha_{l,L_i}$'s are obtained in descending order like (\ref{eq:sorting}). It is noteworthy that $\alpha_{k,K}$ is defined for a fixed $K$ set; however, $\alpha_k$ is the maximum value over all the subsets with cardinality $k$. Suppose that when $z=z^{*}$ and $K=K^{*}$, $\alpha_k$ is achieved in (\ref{eq:alphak}). From the aforementioned definitions and similar argument as in Lemma \ref{lemma:cheapupper}, we have:
\setlength{\mathindent}{0pt}
\par\noindent
\small
\begin{align*}
       \alpha_k = \alpha_{k,K^{*}} \leq \frac{1}{{\binom{k-1}{l-1}}}  \sum_{ \{L_i \subseteq K^{*},\; |L_i|=l\} }  \alpha_{l,L_i}
      \leq  \frac{1}{{\binom{k-1}{l-1}}}  \sum_{ i=1 }^{\binom{k}{l}} \alpha_{l,L_i}.
\end{align*}
\normalsize
The first inequality is from Lemma \ref{lemma:cheapupper}, and the last inequality is from the assumption that $\alpha_{l,L_i}$'s are sorted in descending order.
\qed
\end{proof}
\IncMargin{-1.2em}
\small
\begin{algorithm}[t]
\caption{\label{pick_l_algo_pseudo} Pick-$l$-element algorithm, $1 \leq l < k$ for computing an upper bound of $\alpha_k$}
\begin{algorithmic}[1]
\State Given a matrix $A$, calculate $\alpha_{l,L}$'s for all the subsets $L$, $|L|=l$, via (\ref{opt_prob2}).
\State Sort these $\binom{n}{l}$ different values of $\alpha_{l,L}$'s in descending order like (\ref{eq:sorting}).
\State Compute an upper bound of $\alpha_k$ via (\ref{pick_l}).
\State If the upper bound of $\alpha_k$ is larger than 1, then, set the upper bound to 1. If the upper bound is less than $\frac{1}{2}$, then NSC for $k \in \mathbb{Z}^{+}$ is satisfied.
\end{algorithmic}
\end{algorithm}
\DecMargin{-1.2em}
\normalsize

The steps 2 and 3 in Algorithm \ref{pick_l_algo_pseudo}, which are sorting $\alpha_{l,L}$'s and computing an upper bound of $\alpha_k$ with sorted $\alpha_{l,L}$'s via (\ref{pick_l}), can also be done by solving the following optimization problem without sorting operation:
\setlength{\mathindent}{15pt}
\par\noindent
\small
\begin{align}
\label{optimized_coefficients_without_constraint}
        & \underset{\gamma_i,\; 1 \leq i \leq \binom{n}{l}}{ \text{maximize}}  \;\; \sum_{i=1}^{\binom{n}{l}} \gamma_i \; \alpha_{l,L_{i}}  \nonumber \\
        & \text{subject to}  \;\; 0 \leq \gamma_i \leq \frac{1}{{\binom{k-1}{l-1}}},\; 1\leq i \leq \scriptstyle{\binom{n}{l}},  \nonumber \\
        & \quad\quad\quad\quad\;\; \sum_{i=1}^{\binom{n}{l}} \gamma_i \leq \frac{k}{l}.
\end{align}
\normalsize
Here, we note that $\frac{1}{{\binom{k-1}{l-1}}} \times {\binom{k}{l}} = \frac{k}{l}$. Therefore, for the optimal value, the first $\binom{k}{l}$ largest $\alpha_{l,L_i}$'s are chosen with the coefficient $\frac{1}{{\binom{k-1}{l-1}}}$.

The upshot of the pick-$l$ algorithm is that we can reduce number of operations from $\binom{n}{k}$ enumerations to $\binom{n}{l}$. For example, when $n=300$, $k=20$, and $l=2$, the number of operations is reduced by around $10^{26}$ times. Moreover, as $n$ increases, the reduction rate increases. With the reduced enumerations, we can still have non-trivial upper bounds of $\alpha_k$ through the pick-$l$-element algorithm. We will present the performance of the pick-$l$ algorithm in Section \ref{sec_result} showing that the pick-$l$ algorithm provides better upper bounds than the previous research \cite{d2011testing,juditsky2011verifiable} even when $l=2$. Furthermore, thanks to the pick-$l$ algorithm, we can design a new algorithm based on a branch-and-bound search to calculate $\alpha_k$ by using upper bounds of $\alpha_k$ obtained from the pick-$l$ algorithm. It is noteworthy that the cheap upper bound introduced in Lemma \ref{lemma:cheapupper} can provide upper bounds on $\alpha_{k,K}$ for specific subsets $K$, which enable our branch-and-bound method to calculate $\alpha_k$ or more precise bounds on $\alpha_k$. However, LP relaxation method \cite{juditsky2011verifiable} and SDP method \cite{d2011testing} do not provide upper bounds on $\alpha_{k,K}$ for specific subsets $K$, which overwhelms LP and SDP methods to be used in the branch-and-bound method.

Since we are also interested in $k_{max}$, we introduce the following Lemma \ref{kmax_bound} to bound the maximum recoverable sparsity $k_{max}$.
\begin{lemma}
\label{kmax_bound}
The \textit{maximum recoverable sparsity} $k_{max}$ satisfies
\par\noindent
\small
\begin{align}
\label{lemma:k_max}
    k(\alpha_l) \triangleq \bigg\lceil { l \cdot \frac{1/2}{\alpha_l}} \bigg\rceil - 1 \leq k_{max},
\end{align}
\normalsize
where $\lceil{.}\rceil$ is the ceiling function.
\end{lemma}
\begin{proof}
To prove this lemma, we will show that when $k = \big\lceil { l \cdot \frac{1/2}{\alpha_l}} \big\rceil - 1$, $\alpha_k < \frac{1}{2}$. This can be concluded from the upper bound of $\alpha_k$ given as follows:
\par\noindent
\small
\begin{align}
\label{recoverable_sparsity}
    \alpha_k = \alpha_{k,K^{*}} & \leq \frac{1}{\binom{k-1}{l-1}} \sum_{ \{L_i \subseteq K^{*},\; |L_i|=l\} } \alpha_{l,L_i}\nonumber\\
             & \leq  \frac{\binom{k}{l}}{\binom{k-1}{l-1}} \alpha_l = \alpha_l \cdot \frac{k}{l}.
\end{align}
\normalsize
Note that there are $\binom{k}{l}$ terms in the summation. From (\ref{recoverable_sparsity}), if  $\alpha_l \cdot \frac{k}{l} < \frac{1}{2}$, then  $\alpha_k < \frac{1}{2}$. In other words, if $k < l \cdot \frac{1/2}{\alpha_l}$, then $\alpha_k < \frac{1}{2}$. Since $k$ is a positive integer, when $k = \big\lceil { l \cdot \frac{1/2}{\alpha_l}} \big\rceil - 1$, $\alpha_k < \frac{1}{2}$. Therefore, the maximum recoverable sparsity $k_{max}$ should be larger than or at least equal to $\big\lceil {l \cdot  \frac{1/2}{\alpha_l}} \big\rceil - 1$.
\qed
\end{proof}
It is noteworthy that in \cite[Section 4.2.B]{juditsky2011verifiable}, the authors introduced lower bound on $k$ based on $\alpha_1$, i.e., $k(\alpha_1)$. However, in Lemma \ref{kmax_bound}, we provide a more general result. Furthermore, in Lemma \ref{kmax_bound}, instead of using $\alpha_l$, we can use an upper bound of $\alpha_l$ to obtain the recoverable sparsity $k$; namely, $k(UB(\alpha_l)) = \bigg\lceil { l \cdot \frac{1/2}{UB(\alpha_l)}} \bigg\rceil - 1 \leq k_{max}$, where $UB(\alpha_l)$ represents an upper bound of $\alpha_l$. Since the proof follows the same track as the proof of Lemma \ref{kmax_bound}, we omit the proof.

Finally, we introduce the following proposition to compare our algorithm to LP method \cite{juditsky2011verifiable} theoretically.
\begin{proposition}
\label{thm:LPvsPick1}
For any integer $k\geq 1$, let $\alpha_{k}^{pick1}$ be the upper bound on $\alpha_k$ provided by the pick-$1$-element algorithm according to Lemma \ref{lemma_pickl_upper_bound}.
Let $\alpha_k^{LP}$ be the upper bound on $\alpha_k$ provided by the LP method \cite{juditsky2011verifiable} according to the following definition (namely Equation (4.25) in \cite{juditsky2011verifiable} with $\beta=\infty$) :
\setlength{\mathindent}{0pt}
\par\noindent
\small
\begin{align*}
    \alpha_{k}^{LP} = \underset{Y=[y_1,...,y_n]\in \mathbb{R}^{m\times n}}{ \text{minimize}}\;\bigg\{ \underset{1\leq j \leq n}{ \text{maximize}}\; ||(I-Y^T A)e_j||_{k,1}\; \bigg\},
\end{align*}
\normalsize
where $e_j$ is the standard basis vector with the $j$-th element equal to $1$, and $\|\cdot\|_{k,1}$ stands for the sum of $k$ maximal magnitudes of components of a vector.
Then we have:
\setlength{\mathindent}{15pt}
\par\noindent
\small
\begin{align}
 \alpha_{k}^{pick1} \geq  \alpha_k^{LP}.
\end{align}
\normalsize
\end{proposition}
For readability, we place the proof of Theorem \ref{thm:LPvsPick1} in Appendix A.

The LP method can provide tighter upper bounds on $\alpha_k$ than the pick-$1$-element algorithm, however this comes at a cost of solving a big optimization problem of design dimension $mn$.  When $m$ and $n$ are large, the complexity of computing $\alpha_k^{LP}$ can be prohibitive (please see Table 2).

\section{Optimized Pick-$l$ Algorithm}
\label{Sec2}
We can tighten the upper bound of $\alpha_k$ in the pick-$l$ algorithm by replacing the constant factor $\frac{1}{\binom{k-1}{l-1}}$ in (\ref{pick_l}) with optimized coefficients at the cost of additional complexity, which we call as the \textit{optimized} pick-$l$ algorithm. This optimized pick-$l$ algorithm is mostly useful from a theoretical perspective. In practice, it gives improved but similar performance in calculating the upper bound of $\alpha_k$ to the basic pick-$l$ algorithm described in Section \ref{Sec1}. As a theoretical merit of the \textit{optimized} pick-$l$ algorithm, we can show that as $l$ increases, the upper bound of $\alpha_k$ becomes smaller or stays the same.

The optimized pick-$l$ algorithm provides an upper bound of $\alpha_k$ via the following optimization problem:
\setlength{\mathindent}{15pt}
\par\noindent
\small
\begin{align}
\label{pickl_optimized_coefficients}
        & \underset{\gamma_i,\; 1 \leq i \leq \binom{n}{l}}{ \text{maximize}}                 \;\; \sum_{i=1}^{\binom{n}{l}} \gamma_i \; \alpha_{l,L_{i}} \nonumber \\
        & \text{subject to}  \;\; \gamma_i \geq 0,\; 1 \leq i \leq \scriptstyle{\binom{n}{l}},   \nonumber \\
        & \quad\quad\quad\quad\;\; \sum_{i=1}^{\binom{n}{l}} \gamma_i \leq \frac{k}{l},   \\
        & \quad\quad \sum_{ \{i:\; B \subseteq L_i,\; 1\leq i \leq \binom{n}{l} \} } \gamma_i  \leq \frac{\binom{k-b}{l-b}}{\binom{k-1}{l-1}},
                            \;\;\begin{subarray}{l} {\text{$\forall$ $b \in \mathbb{Z}^{+}$ s.t. $1\leq b \leq l$},} \\ {\text{$\forall$ $B$ with $|B|=b$}} \end{subarray}. \nonumber
\end{align}
\normalsize

In the following lemmata, we show that the optimized pick-$l$ algorithm produces an upper bound of $\alpha_k$ and this bound is tighter than that of the basic pick-$l$ algorithm introduced in (\ref{optimized_coefficients_without_constraint}). The last lemma establishes that as $l$ increases, the upper bound of $\alpha_k$ decreases or stays the same.
\begin{lemma}
\label{pickl_optimized_coefficients_upper_bound}
The optimized pick-$l$ algorithm provides an upper bound of $\alpha_k$.
\end{lemma}
\begin{proof}
The strategy to prove Lemma \ref{pickl_optimized_coefficients_upper_bound} is to show that one feasible solution of (\ref{pickl_optimized_coefficients}) gives an upper bound of $\alpha_k$. Suppose when $K=K^{*}$, $\alpha_k$ is achieved, i.e., $\alpha_k=\alpha_{k,K^{*}}$. For a feasible solution, let us choose $\gamma_{i} = \frac{1}{\binom{k-1}{l-1}}$ when $L_i \subseteq K^*$, and $\gamma_{i} = 0$ otherwise, which we can easily check whether it satisfies the first and second constraints of (\ref{pickl_optimized_coefficients}). For the third constraint, let us check the case when $b=l$ first. For $b=l$, we can choose an arbitrary index set $B$ such that $|B|=b=l$. For the chosen $B$, there is only one $L_i$ such that $B \subseteq L_i$, which is itself, i.e., $B=L_i$. For other chosen $B$'s, it is the same. Hence, the third constraint represents
\par\noindent
\small
\begin{align}
\label{ieq:thrid_const_l}
         \gamma_i \leq \frac{1}{\binom{k-1}{l-1}},\; i=1,2,...,\binom{n}{l}.
\end{align}
\normalsize
For $b=1$, the third constraint represents
\par\noindent
\small
\begin{align}
\label{ieq:third_const}
         \sum_{ \{i:\; B\subseteq L_i,\;1\leq i \leq \binom{n}{l},\; |B|=1 \} } \gamma_i \leq 1.
\end{align}
\normalsize
Note that there are $\binom{n-1}{l-1}$ numbers of $L_i$'s which have an index set $B$ as a subset. Among $\binom{n-1}{l-1}$ numbers of $\gamma_i$'s, only $\gamma_i$'s whose corresponding $L_i$'s are the subsets of $K^{*}$ are $\frac{1}{\binom{k-1}{l-1}}$. Since each element in $L_i$ such that $L_i \subseteq K^{*}$ appears $\binom{k-1}{l-1}$ times in $\{i:\; L_i \subseteq K^{*},\; 1\leq i \leq \binom{n}{l} \}$, the summation of $\gamma_i$, where the corresponding $L_i$'s are the subset of $K^{*}$, becomes $\frac{1}{\binom{k-1}{l-1}}\times \binom{k-1}{l-1} = 1$, which satisfies (\ref{ieq:third_const}). Basically, the third constraint makes that for an index, the summation of coefficients related to the index is limited to 1. In the same way, for $1<b<l$, the chosen $\gamma_i$ is a feasible solution of (\ref{pickl_optimized_coefficients}). From this feasible solution, we have $\frac{1}{{\binom{k-1}{l-1}}} \sum_{ \{i:\; L_i \subseteq K^{*},\; |L_i|=l\} } \alpha_{l,L_i}$ for the optimal value, which is an upper bound of $\alpha_k$ as shown in (\ref{recoverable_sparsity}).
\qed
\end{proof}

\begin{lemma}
\label{pickl_optimized_coefficients_vs_the_pickl_element_algo}
The optimized pick-$l$ algorithm provides a tighter, or at least the same, upper bound of $\alpha_k$ than the basic pick-$l$ algorithm introduced in (\ref{optimized_coefficients_without_constraint}).
\end{lemma}
\begin{proof}
We will show that the optimization problem (\ref{optimized_coefficients_without_constraint}) is a relaxation of (\ref{pickl_optimized_coefficients}). As in the proof of Lemma \ref{pickl_optimized_coefficients_upper_bound}, for $b=l$, the third constraint of (\ref{pickl_optimized_coefficients}) represents (\ref{ieq:thrid_const_l}), which is involved in the first constraint of (\ref{optimized_coefficients_without_constraint}). Since the third constraint of (\ref{pickl_optimized_coefficients}) considers other $b$ values such that $1 \leq b < l$, (\ref{pickl_optimized_coefficients}) has more constraints than (\ref{optimized_coefficients_without_constraint}).
Therefore, the optimized pick-$l$ algorithm, which is (\ref{pickl_optimized_coefficients}), provides a tighter or at least the same upper bound than the basic pick-$l$ algorithm.
\qed
\end{proof}

\begin{lemma}
\label{lemma3_various_coefficients}
The optimized pick-$l$ algorithm provides a tighter or at least the same upper bound than the optimized pick-$p$ algorithm when $l > p$.
\end{lemma}
\begin{proof}
We can upper bound the objective function of (\ref{pickl_optimized_coefficients}) by using (\ref{eq:cheapbound}) as follows:
\par\noindent
\small
\begin{align}
\label{lemma3_optimized_coefficients_proof2}
        & \underset{\gamma_i,\; 1 \leq i \leq \binom{n}{l}}{\text{maximize}} \;\; \frac{1}{\binom{l-1}{p-1}}  \sum_{i=1}^{\binom{n}{l}} \gamma_i \sum_{ \{j :\; P_j \subset L_i,\; |P_j|=p \} }  \alpha_{p,P_{j}}  \nonumber \\
        & \text{subject to} \;\; \gamma_i \geq 0,\; 1 \leq i \leq \scriptstyle{\binom{n}{l}},   \nonumber \\
        & \quad\quad\quad\quad\;\; \sum_{i=1}^{\binom{n}{l}} \gamma_i \leq \frac{k}{l},    \\
        & \quad\quad \sum_{ \{i:\; B \subseteq L_i,\; 1\leq i \leq \scriptstyle{\binom{n}{l}} \} } \gamma_i  \leq \frac{\binom{k-b}{l-b}}{\binom{k-1}{l-1}},
                            \;\;\begin{subarray}{l} {\text{$\forall$ $b \in \mathbb{Z}^{+}$ s.t. $1\leq b \leq l$},} \\ {\text{$\forall$ $B$ with $|B|=b$}} \end{subarray}. \nonumber
\end{align}
\normalsize
Note that in the objective function of (\ref{lemma3_optimized_coefficients_proof2}), each $\alpha_{p,P_j},\; 1\leq j\leq \binom{n}{p}$, appears $\binom{n-p}{l-p}$ times. Let us define
\par\noindent
\small
\begin{align*}
    \gamma_j^{'} \triangleq \frac{1}{\binom{l-1}{p-1}} \sum_{ \{i:\; P_j \subset L_i,\; 1\leq i \leq \binom{n}{l}  \} } \gamma_i.
\end{align*}
\normalsize
We can relax (\ref{lemma3_optimized_coefficients_proof2}) to the following problem, which turns out to be the same as the optimized pick-$p$ algorithm:
\par\noindent
\small
\begin{align}
\label{lemma3_optimized_coefficients_proof3}
        & \underset{\gamma_j^{'},\; 1 \leq j \leq \binom{n}{p}}{\text{maximize}}                \;\;  \sum_{j=1}^{\binom{n}{p}} \gamma_j^{'} \; \alpha_{p,P_{j}}  \nonumber \\
        & \text{subject to}  \;\; \gamma_j^{'} \geq 0,\; 1 \leq j \leq \scriptstyle{\binom{n}{p}},   \nonumber \\
        & \quad\quad\quad\quad\;\; \sum_{j=1}^{\binom{n}{p}} \gamma_j^{'} \leq \frac{k}{p},    \\
        & \quad\quad \sum_{ \{j:\; B \subseteq P_j,\; 1\leq j \leq \binom{n}{p} \} } \gamma_j^{'}  \leq  \frac{ \binom{k-b}{p-b} }{ \binom{k-1}{p-1} },
                            \;\;\begin{subarray}{l} {\text{$\forall$ $b \in \mathbb{Z}^{+}$ s.t. $1\leq b \leq p$},} \\ {\text{$\forall$ $B$ with $|B|=b$}} \end{subarray}. \nonumber
\end{align}
\normalsize
The relaxation is shown by checking the constraints. The first constraint of (\ref{lemma3_optimized_coefficients_proof3}) is trivial to obtain. For the second constraint, we can obtain the second constraint of (\ref{lemma3_optimized_coefficients_proof3}) from the following relations:
\par\noindent
\small
\begin{align*}
         \sum_{j=1}^{\binom{n}{p}} \gamma_j^{'}
         & = \sum_{j=1}^{\binom{n}{p}} \frac{1}{\binom{l-1}{p-1}} \sum_{ \substack{ \{i:\; P_j \subset L_i, \; 1\leq i \leq \binom{n}{l} \}} } \gamma_i \\
         & = \frac{1}{\binom{l-1}{p-1}} \binom{l}{p} \sum_{i=1}^{\binom{n}{l}} \gamma_i \\
         & \leq \frac{1}{\binom{l-1}{p-1}} \binom{l}{p} \frac{k}{l} = \frac{k}{p},
\end{align*}
\normalsize
where the second equality is obtained from the fact that $\gamma_i$, which is a coefficient of $\alpha_{l,L_i}$, appears $\binom{l}{p}$ times in $\sum_{j=1}^{\binom{n}{p}} \sum_{ \substack{ \{i:\; P_j \subset L_i \}} } \gamma_i$. The final inequality is from the second constraint of (\ref{lemma3_optimized_coefficients_proof2}). The third constraint in (\ref{lemma3_optimized_coefficients_proof3}) can be deduced from the following inequality:
\par\noindent
\small
\begin{align*}
        & \sum_{ \{ j:\; B \subseteq P_j,\; 1\leq j \leq \binom{n}{p}\} } \gamma_j^{'} \\
        & = \frac{1}{\binom{l-1}{p-1}} \sum_{ \{j:\; B \subseteq P_j,\; 1\leq j \leq \binom{n}{p}\} }  \sum_{ \{i :\; P_j \subset L_i,\; 1\leq i \leq \binom{n}{l}  \} } \gamma_i   \nonumber \\
        & = \frac{1}{\binom{l-1}{p-1}}  \frac{ \binom{n-b}{p-b} \binom{n-p}{l-p} }{ \binom{n-b}{l-b} } \sum_{ \{i :\; B \subset L_i,\; 1\leq i \leq \binom{n}{l} \} } \gamma_i   \nonumber \\
        & \leq \frac{1}{\binom{l-1}{p-1}}  \frac{ \binom{n-b}{p-b} \binom{n-p}{l-p} }{ \binom{n-b}{l-b} }  \frac{\binom{k-b}{l-b}}{\binom{k-1}{l-1}}, \; 1 \leq b \leq p \nonumber \\
        & = \frac{ \binom{k-b}{p-b} }{ \binom{k-1}{p-1} }, \;1 \leq b \leq p,
\end{align*}
\normalsize
where the second equality is from the fact that for a fixed $P_j$, there are $\binom{n-p}{l-p}$ numbers of $L_i$'s, where $P_j \subset L_i$, $i=1,...,\binom{n}{l}$; for a fixed $B$, there are $\binom{n-b}{p-b}$ numbers of $P_j$'s, where $B \subset P_j$, $j=1,...,\binom{n}{p}$, and $\binom{n-b}{l-b}$ numbers of $L_i$'s, where $B \subset L_i$, $i=1,...,\binom{n}{l}$. Since (\ref{lemma3_optimized_coefficients_proof3}) is obtained from the relaxation of (\ref{lemma3_optimized_coefficients_proof2}), the optimal value of (\ref{lemma3_optimized_coefficients_proof3}) is larger or equal to the optimal value of (\ref{lemma3_optimized_coefficients_proof2}). (\ref{lemma3_optimized_coefficients_proof3}) is just the optimized pick-$p$ algorithm. Thus, when $l > p$, the optimized pick-$l$ algorithm provides a tighter or at least the same upper bound than the optimized pick-$p$ algorithm.
\qed
\end{proof}

By using larger $l$ in the pick-$l$ algorithm, we can obtain a tighter upper bound of $\alpha_k$. However, for a certain $l$, we need to enumerate $\binom{n}{l}$ possibilities,  and this becomes infeasible when $l$ is large.  Moreover,  when $l<k$, the pick-$l$ algorithm only gives an upper bound of $\alpha_{k}$, instead of an exact value of $\alpha_{k}$. There is, however, a need to find tighter bounds on $\alpha_{k}$ , or to even find the exact value of $\alpha_{k}$, when $k$ is too large for $\binom{n}{k}$ enumerations of exhaustive search \cite{xu2011Compressive,firooz2010network,coates2001network}.  To this end, we propose a new branch-and-bound tree search algorithm to  find tighter bounds on $\alpha_{k}$ than Lemma \ref{lemma_pickl_upper_bound} provides, or to even find the exact $\alpha_{k}$ .   Our branch-and-bound tree search algorithm is enabled by the pick-$l$ algorithms introduced in Sections \ref{Sec1} and \ref{Sec2}.

\section{Tree Search Algorithm}
\label{Sec3}
To find the index set $K^{*}$ which leads to the maximum $\alpha_{k,K}$ (among all possible index set $K$'s), the Tree Search Algorithm (TSA) performs a best-first branch-and-bound search \cite{brassard1996fundamentals} over a tree structure representing different subsets of $\{1,2, ..., n\}$.  In its essence, for each subset $J$ with cardinality no bigger than $k$, TSA calculates an upper bound of  $\alpha_{k,K}$, which is valid for any set $K$ (with cardinality $k$) such that $J \subseteq K$.
If this upper bound is smaller than a lower bound of $\alpha_k$, TSA will not further explore any of $J$'s supersets, leading to reduced average-case computational complexity. For simplicity,  we will describe the TSA based on pick-$1$-element algorithm, simply called 1-Step TSA. However, we remark we can also extend the TSA to be based on pick-$l$-element ($l\geq 2$) algorithm, by calculating upper bounds of $\alpha_{k,K}$ based on the results of the pick-$l$-element algorithm.

\subsection{Tree structure}
A tree node $J$ represents an index subset of $\{1, ..., n\}$ such that $|J| \leq k$.  We have the following rule:
\begin{enumerate}
\item[\textbf{[R1]}] A parent node is a subset of each of its child node(s).
\end{enumerate}

A node that has no child is referred to as a \textit{leaf node}. We call the cardinality of the index set corresponding to $J$ as $J$'s \textit{height}.  The tree structure follows the ``legitimate order'', which ensures that any new index in the child node is bigger than the indices of its parent node.
\begin{enumerate}
\label{def:legitimate}
\item[\textbf{[R2]}]
``Legitimate order'' - Let $P$ and $C$ denote the parent node, and the child node. Then, any index in $P$ must be smaller than any index in $C \setminus P$.
\end{enumerate}
Fig. \ref{tree_draw} illustrates this rule in a tree with $k= 2$ and $n=3$.
\begin{figure}[h]
    \centering
    \includegraphics[scale=0.7]{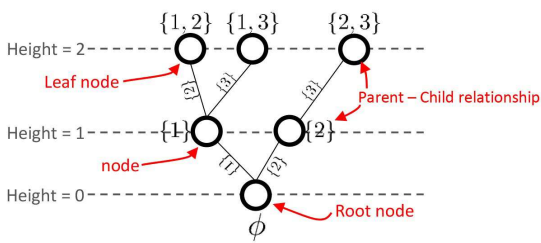}
    \caption{\small{A tree structure following the legitimate order for $k = 2$ and $n=3$.}}
    \label{tree_draw}
\end{figure}

\subsection{Basic idea of a branch-and-bound approach for calculating $\alpha_{k}$}
We use a branch-and-bound approach over the tree structure to calculate $\alpha_{k}$. This method maintains a lower bound on $\alpha_{k}$ (how to maintain this lower bound will be explained in
Subsection \ref{subsec:bestfirstimplementation}).
When the algorithm explores a tree node $J$,
the algorithm calculates an upper bound $B(J)$, which is no smaller
than $\alpha_{k,K}$ for any child node $K$ (with cardinality $k$) of node $J$.
If $B(J)$ is smaller than the lower bound on $\alpha_{k}$, then the algorithm will not explore the child nodes of the tree node $J$.

In our algorithm, we calculate $B(J)$ as
\par \noindent\small
\begin{align}
\label{eq:upperboundJ}
    B(J) = \alpha_{j,J} + {\sum_{i=1}^{t} \alpha_{1,\{i+max(J)\}}},
\end{align}
\normalsize
where $j+t=k$, $max(J)$ represents the largest index in $J$, and $\alpha_{1,\{1\}} \geq \alpha_{1,\{2\}} \geq ... \geq \alpha_{1,\{n\}}$. We obtain this descending order by permuting the columns of the sensing matrix $A$ in descending order of $\alpha_{1,\{i\}}$'s as the pre-computation step of TSA. For example, in Fig. \ref{tree_draw}, for $k=2$, $B(\{1\})=\alpha_{1,\{1\}} + \alpha_{1,\{2\}}$. In order to justify that $B(J)$ is an upper bound of $\alpha_{k,K}$ for all node $K$ such that $J \subseteq K$, we provide the following lemma.
\begin{lemma}
\label{lemma:eq:upperboundJ}
Given $\alpha_{1,\{1\}} \geq \alpha_{1,\{2\}} \geq ... \geq \alpha_{1,\{n\}}$, $B(J) = \alpha_{j,J} + \sum_{i=1}^{t} \alpha_{1,\{i+max(J)\}}$, where $j+t = k$, and $max(J)$ represents the largest index in $J$, is an upper bound of $\alpha_{k,K}$ for all nodes $K$ such that $J \subseteq K$. 
\end{lemma}
\begin{proof}
For any subset $K$ such that $J \subseteq K$, we can write $\alpha_{k,K} = \alpha_{j+t,\{J \cup T\}}$, where $j+t=k$ and $T = K \setminus J$. Then, following exactly the same line of argument as in the proof of Lemma \ref{lemma:cheapupper}, we have 
$$\alpha_{k,K} \leq \alpha_{j,J} + \alpha_{t,T},$$
and $\alpha_{t,T}$ is no larger than $\sum_{j \in T}^{t} \alpha_{1,\{j\}}$. Finally, since $\alpha_{1,\{i\}}$'s are sorted in the descending order, $\sum_{j \in T} \alpha_{1,\{j\}} \leq \sum_{i=1}^{t} \alpha_{1,\{i+max(J)\}}$. Note that, due to the legitimate order \textbf{[R2]}, the smallest element of the index set $T$ is no less than $1+max(J)$. In conclusion, for all nodes $K$ such that $J \subseteq K$, $B(J)$ becomes an upper bound of $\alpha_{k,K}$.
\qed
\end{proof}

\subsection{Best-first tree search strategy}
\label{subsec:bestfirstimplementation}
TSA adopts a best-first tree search strategy for the branch-and-bound approach. We first describe a basic version of the best-first tree search strategy, and then introduce two enhancements to this strategy in the next subsection.

In its basic version, TSA starts with a tree having only the root node, and sets the global lower bound of $\alpha_{k}$ as $0$. In each iteration, TSA selects a leaf tree node $J$ with the largest $B(J)$, and expands the tree by adding the child nodes of $J$ to the tree.  For each of these newly added child nodes, say $Q$,  TSA then calculates the upper bound $B(Q)$ in (\ref{eq:upperboundJ}).  Note that if a newly added child node $Q$ has $k$ elements, TSA will calculate $\alpha_{{k,Q}}$, \emph{which is a lower bound on $\alpha_k$}. For this $k$-element $Q$, if the newly calculated $\alpha_{{k,Q}}$ is bigger than the global lower bound of $\alpha_{k}$, TSA will set the global lower bound equal to $\alpha_{{k,Q}}$. TSA will terminate if a leaf tree node $J$ has the largest $B(J)$ among all the leaf nodes, and that $B(J)$ is no bigger than the global lower bound on $\alpha_k$. 

From standard theories of the branch-and-bound approach, this TSA will output the exact $\alpha_{k}$. Also, in this process, the global lower bound will keep increasing until it is equal to an upper bound of $\alpha_{k}$ (the largest $B(J)$ among leaf nodes).

%

\subsection{Two enhancements}
We incorporate two novel features to TSA in order to reduce the computational complexity. Firstly, when TSA attaches a new node $Q$ to a node $J$ in the tree structure, TSA computes $B(Q)$ as (\ref{eq:upperboundQ}):
\par\noindent\small
\begin{align}
\label{eq:upperboundQ}
    B(Q) = \alpha_{j,J} +\alpha_{1, Q\setminus J}+\sum_{i=1}^{t} \alpha_{1,\{i+max(Q)\}},
\end{align}
\normalsize
where $j+t+1=k$, $max(Q)$ represents the largest index in $Q$, and $\alpha_{1,\{1\}} \geq \alpha_{1,\{2\}} \geq ... \geq \alpha_{1,\{n\}}$. Thus, without calculating  $\alpha_{j+1, Q}$ (which involves higher computational complexity), we can still have $B(Q)$ as an upper bound of $\alpha_{k,K}$ for any child node $K$ (with cardinality $k$) of the node $Q$.

Secondly, when TSA adds a new node $Q$ as the child of node $J$ in the tree structure (assuming $\alpha_{j, J}$ has already been calculated), TSA does not need to add all of $J$'s child nodes to the tree at the same time.  Instead, TSA only adds the node $J$'s unattached child node $Q$ with the largest $B(Q)$ as defined in (\ref{eq:upperboundQ}). Namely, the index $Q\setminus J$ is no bigger than the index $Q'\setminus J$,  where $Q'$ is any unattached child of the node $J$. We note that $B(Q)$ is an upper bound on $B(Q')$ (according to (\ref{eq:upperboundQ})) for any other unattached child node $Q'$ of the node $J$.
Thus, for any child node $K$ (of cardinality $k$) of node $J$'s unattached child nodes, $B(Q)$ is still an upper bound of $\alpha_{{k,K}}$.

Algorithm \ref{TSA_algo} shows detailed steps of TSA, based on the pick-$1$-element algorithm (namely, $l=1$, $1$-Step TSA). In the  description, we define ``expanding the tree from a node $J$'' as follows:
\begin{enumerate}
\item[\textbf{[R3]}]``Expanding the tree from a node $J$'' - Attaching a new node $Q$ to the node $J$, where $B(Q)$ is the largest value defined as (\ref{eq:upperboundQ}) among the node $J$'s all the unattached child nodes.
\end{enumerate}

\subsection{Advantage of the tree search algorithm}
Due to the nature of the branch-and-bound approach, we can obtain a global upper bound and a global lower bound of $\alpha_k$ while TSA runs. As the number of iterations increases in TSA, we can obtain tighter and tighter upper bounds on $\alpha_{k}$, which is the largest $B(\cdot)$ among leaf nodes. By using the global upper bound of $\alpha_k$, we can obtain a lower bound of the recoverable sparsity $k$ via Lemma \ref{kmax_bound}. Thus, even if the complexity of TSA is too high to finish in a timely manner, we can still obtain a lower bound on the recoverable sparsity $k$ by early terminating TSA.

We note that the methods based on LP \cite{juditsky2011verifiable} and SDP \cite{d2011testing} also provide upper bounds on $\alpha_{k}$. However, they are unable to determine upper bounds of $\alpha_{k,K}$, which is for a specific index set $K$.
This prevents the use of LP and SDP methods in our branch-and-bound method for computing $\alpha_{k}$. %


\begin{algorithm}[t!]
\LinesNumbered
  \caption{Tree search algorithm based on the pick-$1$-element algorithm ($1$-Step TSA)}
  \label{TSA_algo}
  \SetAlgoLined
  \SetKwFor{Loop}{Loop}{}{EndLoop}
{\footnotesize
   \KwIn{ $A \in \mathbb{R}^{m \times n}$, $k$, $l \leftarrow 1$ \Comment{$1$-Step TSA, i.e., $l=1$} }
   \KwOut{ $\alpha_k$ }
   \BlankLine
   \nonl \textbf{$\triangleright \;\; \text{Pre-computation:}$ } \par
   compute $\alpha_{l,\{i\}}$ for $i=1,...,n$ via (\ref{opt_prob2}) \par
   permute columns of $A$ in descending order of $\alpha_{1,\{i\}}$'s, so that $\alpha_{1,\{1\}} \geq ...\geq \alpha_{1,\{n\}}$\par
   \BlankLine
   \nonl \textbf{$\triangleright \;\; \text{Tree expansion:}$ } \par
   start with root node $\varnothing$, where $B(\varnothing)=\sum_{i=1}^{k}\alpha_{1,\{i\}}$, in a tree structure $\Upsilon$\par
   \Loop{}{
        $J$ $\leftarrow$ a node that has the largest $B(\cdot)$ among all the leaf nodes in $\Upsilon$ \par
        $j$ $\leftarrow$ $|J|$  \par
        \uIf { $\alpha_{j,J}$ is not calculated }
        {
            compute $\alpha_{j,J}$ via (\ref{opt_prob2}) and update $B(J)$ via (\ref{eq:upperboundJ}) \par
            expand $\Upsilon$ from the parent of $J$  \Comment{See \textbf{[R3]}} \par
        }
        \Else
        {
            \uIf { $j = k$ }
            {
                $\alpha_k$ $\leftarrow$ $B(J)$ \par
                break \par
            }
            \Else
            {
                expand $\Upsilon$ from $J$ \Comment{See \textbf{[R3]}}\par
            }
        }
   }
}%
\end{algorithm}

\section{Numerical Experiments}
\label{sec_result}
We conducted extensive simulations to compute $\alpha_k$ and its upper/lower bounds using the pick-$l$ algorithms and TSA. In this section, we call the pick-$l$ algorithms introduced in Section \ref{Sec1} and \ref{Sec2} as simply the (basic) pick-$l$ and the optimized pick-$l$ algorithms respectively.

For same matrices, we compared our methods with LP relaxation \cite{juditsky2011verifiable} approach and SDP method \cite{d2011testing}. We assessed the computational complexity in terms of execution time of the algorithms.\footnote{We conducted our experiments on HP Z220 CMT with Intel Core i7-3770 dual core CPU @3.4GHz clock speed and 16GB DDR3 RAM, using Matlab (R2013b) on Windows 7.}
In addition, we carried out numerical experiments to demonstrate the computational complexity of TSA empirically.

For LP method in \cite{juditsky2011verifiable} and SDP method in \cite{d2011testing}, we used the Matlab codes\footnote{LP method from \url{http://www2.isye.gatech.edu/~nemirovs/} and SDP method from \url{http://www.di.ens.fr/~aspremon/NSPcode.html}.} provided by the authors. Consistent with previous research, we used CVX \cite{cvx} - a package for specifying and solving convex programs - for the SDP method, and MOSEK \cite{mosek} - a commercial LP solver - for the LP method. In our own algorithms, we used MOSEK to solve (\ref{opt_prob2}). Also, to be consistent with the previous research, matrices were generated from the Matlab code provided by the authors of \cite{d2011testing} at \url{http://www.di.ens.fr/~aspremon/NSPcode.html}. For valid bounds, we rounded down lower bounds on $\alpha_k$ and exact $\alpha_k$, and rounded up upper bounds on $\alpha_k$ to the nearest hundredth.

\subsection{Performance comparison}
Firstly, we considered Gaussian matrices and partial Fourier matrices sized from $n=40$ to $n=6144$. We chose $n=40$ so that our results can be compared with the simulation results in \cite{d2011testing}.

\subsubsection{Low-dimensional sensing matrices}
\underline{\textbf{Sensing matrices with $n=40$:}} We considered sensing matrices of row dimension $m=0.5n$, $0.6n$, $0.7n$, $0.8n$, where $n=40$. For every matrix size, we randomly generated $10$ different realizations of Gaussian and partial Fourier matrices. So in total we used 80 different $n=40$ sensing matrices for the numerical experiments in Tables \ref{tbl:Gaussian} and \ref{tbl:Fourier}. We normalized all of the matrix columns so that they have a unit $\ell_2$-norm. The entries of Gaussian matrices were i.i.d standard Gaussian $\mathcal{N}(0,1)$. The partial Fourier matrices had $m$ rows randomly draw from the full Fourier matrices. We compared our algorithms - pick-$1$-element, pick-$2$-element, pick-$3$-element and TSA - to LP and SDP methods. For readability, we place the numerical results for these small sensing matrices in Appendix B.

For each matrix size and type, we increased $k$ from $1$ to $5$ in unit steps. Tables \ref{tbl:Gaussian} (a) and \ref{tbl:Fourier} (a) show the \textit{median} values of $\alpha_k$. (To be consistent with the previous research \cite{d2011testing}, in which the authors used the median value of $\alpha_k$ to compare the SDP method with the LP method, we provided the median values obtained from 10 random realizations of sensing matrix.) From the median value of $\alpha_k$, we obtained the recoverable sparsity $k_{max}$ such that $\alpha_{k_{max}}< 1/2$ and $\alpha_{k_{max}+1} > 1/2$. In addition, we calculated the arithmetic mean of $k_{max}$'s. For the arithmetic mean, we obtained each $k_{max}$ from  each random realization, and computed the arithmetic mean of ten $k_{max}$'s. Compared with LP and SDP methods, we obtained bigger or at least the same recoverable sparsity $k_{max}$ by using pick-$2$, pick-$3$ and TSA. It is noteworthy that we obtained the exact $\alpha_k$ for $k=1,2,...,5$ by using TSA, while LP and SDP methods only provided the exact $\alpha_k$ for $k=1$. We observed that $\alpha_k < 1/2$ but the upper bound of $\alpha_k > 1/2$ holds true in several cases, e.g.,  $\alpha_5$ in $32 \times 40$ Gaussian matrices, $\alpha_4$ in $28 \times 40$ Gaussian matrices, $\alpha_3$ in $24 \times 40$ Gaussian matrices, $\alpha_3$ in $20 \times 40$ partial Fourier matrices, and $\alpha_4$ in $24 \times 40$ partial Fourier matrices. Additionally, this can also be established by the arithmetic mean of $k_{max}$ in Tables \ref{tbl:Gaussian} (a) and \ref{tbl:Fourier} (a). 

To compare the computational complexity, we calculated the \textit{geometric mean} of the algorithms' execution time, to avoid biases for the average. Tables \ref{tbl:Gaussian} (b) and \ref{tbl:Fourier} (b) list the average execution time. We also ran the Exhaustive Search Method (ESM) to find $\alpha_k$,  and compared its execution time with that of TSA. In calculating $\alpha_5$, on average, 3-Step TSA reduced the computational time by around 86 times for $20 \times 40$ Gaussian matrices, and by 94 times for $20 \times 40$ partial Fourier matrices, compared to ESM. For $32 \times 40$ Gaussian matrix and partial Fourier matrix, the speedup compared to the best $l$-Step TSA, $l=1,2,3$, becomes around 1760 times and 182 times respectively. We observed that when $m/n = 0.5$ e.g. $20 \times 40$ sensing matrices, in general, the $3$-step TSA provides the fastest result for $k = 5$. On the other hand, for $m/n = 0.8$ (e.g. $32 \times 40$ case), the $2$-Step TSA is the quickest in finding an exact $\alpha_k$ for $k = 5$; however, for $k > 5$, the fastest $l$-step TSA cannot be determined from either experiments or theory.

\underline{\textbf{Sensing matrices with $n=256$:}} We assessed the performance of the pick-$l$ algorithm for sensing matrices with $n=256$. We carried out numerical experiments on $128 \times 256$ Gaussian matrices in Fig. \ref{fig:pickVSlp} (a) and $64 \times 256$ partial Fourier matrices in Fig. \ref{fig:pickVSlp} (b). Here, for $10$ sensing matrices, we obtained the median value of upper bounds of $\alpha_k$ using the pick-$l$ algorithm and compared the result with LP relaxation method \cite{juditsky2011verifiable}. We omitted SDP method \cite{d2011testing} from this experiment due to its very high computational complexity. For the pick-$3$ algorithm in Fig. \ref{fig:pickVSlp} (a), we calculated an upper bound of $\alpha_3$ via TSA, and used this result to calculate upper bounds of $\alpha_k$, $k =3,4, ...,8$ via (\ref{recoverable_sparsity}). Fig. \ref{fig:pickVSlp} (a) and (b) demonstrate that, with an appropriate choice of $l$, the upper bound of $\alpha_k$ obtained via the pick-$l$ algorithm can be tighter than that from the LP relaxation method. For example, for $128 \times 256$ Gaussian matrices, LP relaxation often determines the maximum recoverable sparsity as $5$, while the pick-$2$ algorithm improves it to $6$. In the pick-$3$ algorithm, the maximum recoverable sparsity is $7$ ($\alpha_7=0.49$). For $64 \times 256$ partial Fourier matrices, the maximum recoverable sparsity from LP relaxation and the pick-$2$ algorithm are 3 and 4 respectively.
\begin{figure}[t]
    \centering
    \subfloat[$128 \times 256$ Gaussian matrices]{ \includegraphics[scale=0.22]{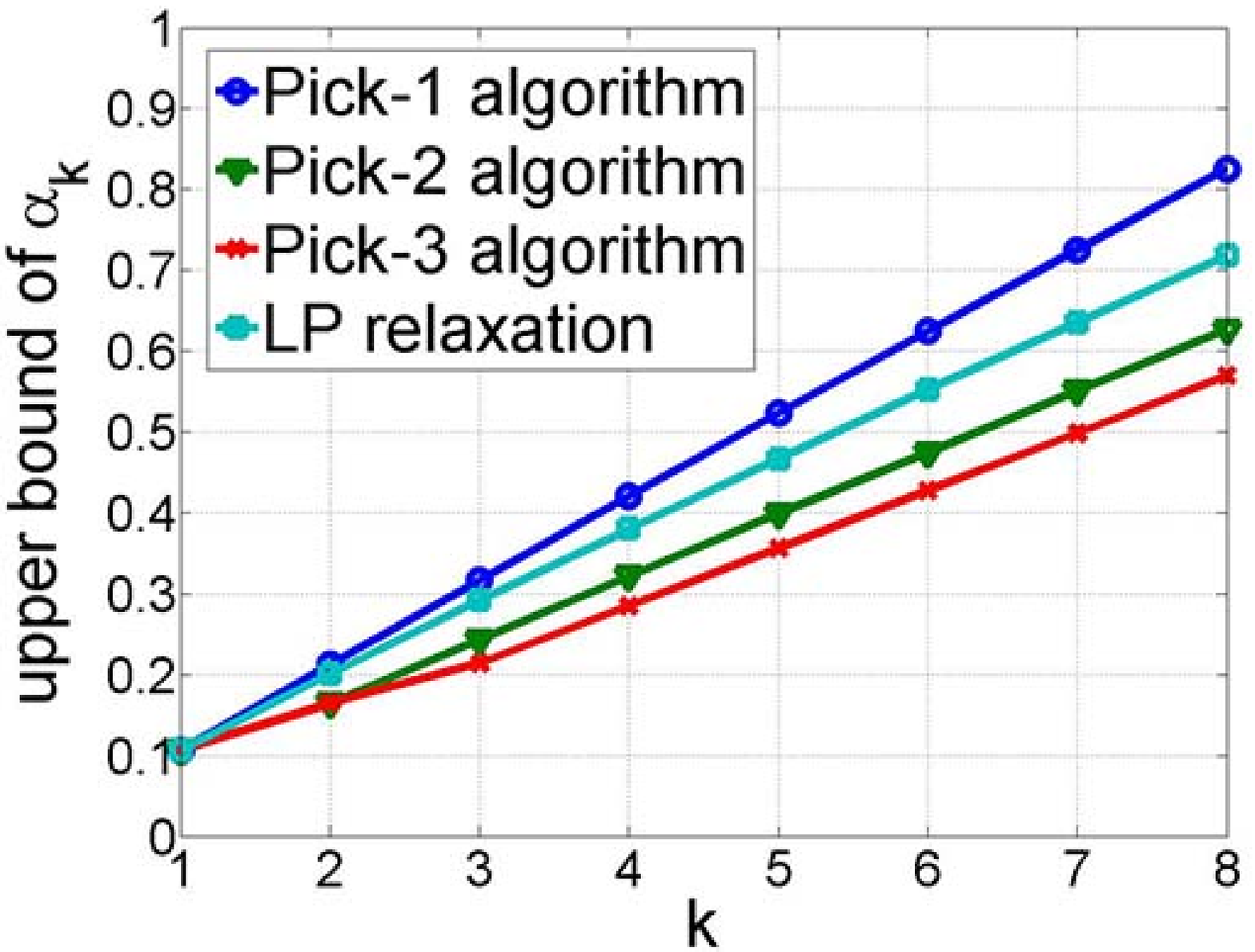}}
    \subfloat[$64 \times 256$ Partial Fourier matrices]{ \includegraphics[scale=0.22]{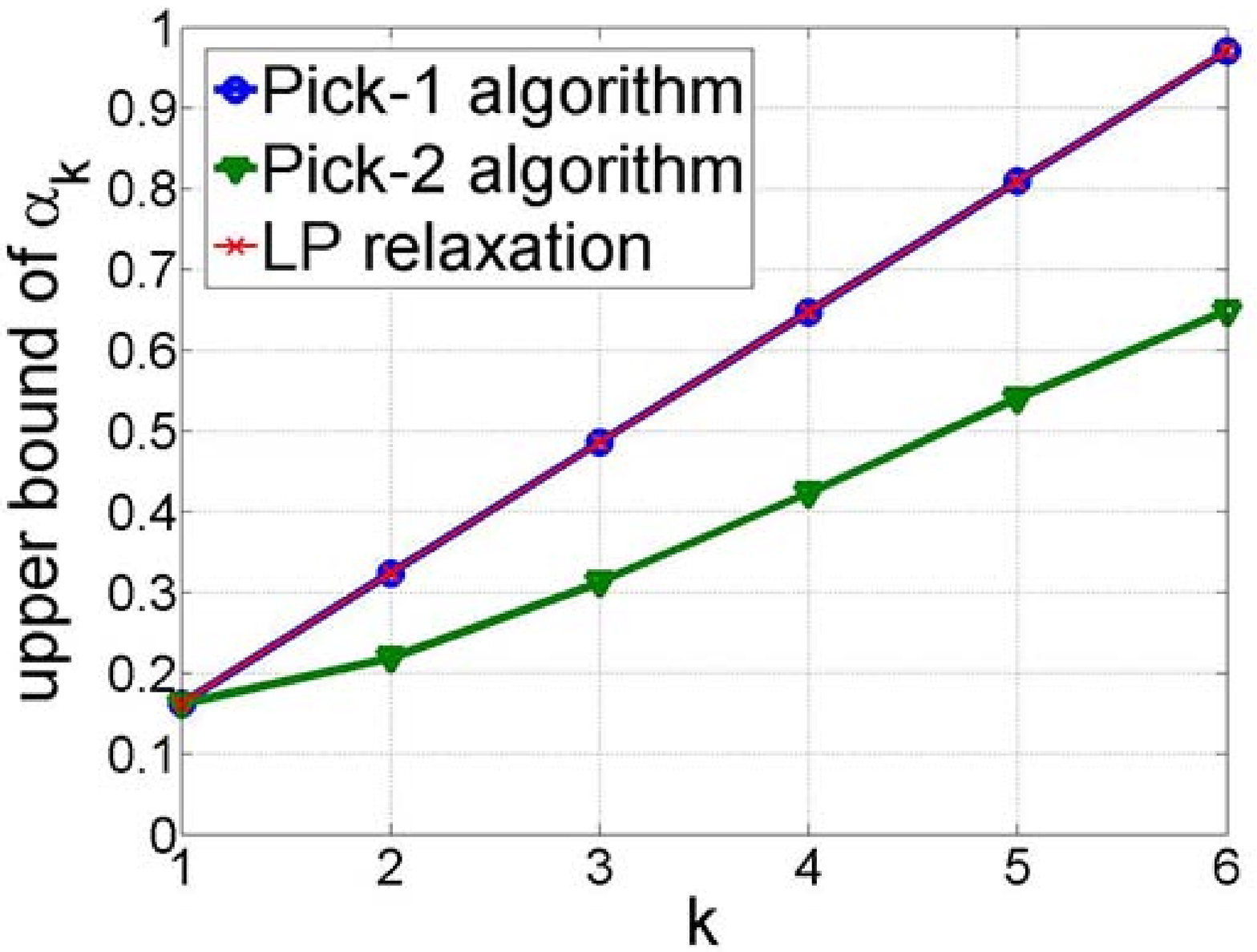}}
    \caption{\small{Median upper bounds of $\alpha_k$ from the pick-$l$ algorithm and the LP relaxation method.}}
    \label{fig:pickVSlp}
\end{figure}

\underline{\textbf{Sensing matrices with $n=512$:}} We further conducted numerical experiments on Gaussian sensing matrices with $n=512$. The simulation results in Table \ref{tbl:Gaussian_LD_512} clearly demonstrate that the pick-2 algorithm provides larger lower bound on the recoverable sparsity $k$ than the LP method \cite{juditsky2011verifiable}. Especially, when Gaussian sensing matrix is $410 \times 512$, the lower bound on $k$ obtained from the pick-2 algorithm is almost twice larger than that of the LP method. 

\begin{table}[t!]
\centering
\caption{\label{tbl:Gaussian_LD_512} Lower bound on $k$ and execution time (Gaussian Matrix with $n=512$)}
\subfloat[Lower bound on $k$]{
\begin{threeparttable}
\setlength{\tabcolsep}{15.0pt}
{\small
    \begin{tabular}{cccc}
        \hline
        $ \text{matrix}\; A$  & Pick-$1$ & Pick-$2$ & LP\tnote{a} \\
        \hline
        $102 \times 512$  & 2   & 3    & 2    \\
        $205 \times 512$  & 5   & 7    & 5    \\
        $307 \times 512$  & 10  & 17   & 10   \\
        $410 \times 512$  & 14  & 27   & 14   \\
        \hline
    \end{tabular}
}
\end{threeparttable}
}\\
\subfloat[Execution time (Unit: second)]{
\centering
\setlength{\tabcolsep}{13.0pt}
\begin{threeparttable}
{\small
    \begin{tabular}{cccc}
        \hline
        $ \text{matrix}\; A $  & Pick-$1$ & Pick-$2$ &
        LP\tnote{a} \\
        \hline
        $102 \times 512$    & 53.7   & 2.96e4  & 50.8    \\
        $205 \times 512$    & 114.8  & 6.36e4  & 105.1   \\
        $307 \times 512$    & 309.7  & 1.19e5  & 333.0   \\
        $410 \times 512$    & 133.1  & 5.03e4  & 510.0   \\
        \hline
    \end{tabular}
}
\end{threeparttable}
}\\[-15pt]
\begin{threeparttable}
\begin{tabular*}{\textwidth}{c}
\end{tabular*}
\end{threeparttable}
\begin{tablenotes}[flushleft]
    \item [a] {\scriptsize Linear Programming \cite{juditsky2011verifiable}} \\
\end{tablenotes}
\end{table}

\subsubsection{High-dimensional sensing matrices}
\underline{\textbf{Sensing matrix with $n \geq 1024$:}} We conducted numerical experiments for Gaussian sensing matrices with $n$ from $1024$ to $6144$. We show these numerical experiments in Tables \ref{tbl:Gaussian_HD} and \ref{tbl:Gaussian_HD2}, where we calculated the lower bound on the recoverable sparsity $k$ and obtained the corresponding execution time. The SDP method \cite{d2011testing} was not applicable in these experiments due to its very high computational complexity. In Table \ref{tbl:Gaussian_HD}, we ran TSA for 1 day (24 hours) and obtained an upper bound of $\alpha_2$, denoted by $UB(\alpha_2)$. With the upper bound of $\alpha_2$, we obtained a lower bound of $k$, denoted by $k(UB(\alpha_2))$, via Lemma \ref{kmax_bound}. Our numerical results in Tables \ref{tbl:Gaussian_HD} and \ref{tbl:Gaussian_HD2} clearly show that our pick-$l$ algorithm outperforms the LP method in recoverable sparsity $k$ or execution time. We note that although our pick-$1$-element algorithm provides the same recoverable sparsity $k$ as the LP method \cite{juditsky2011verifiable} in Tables \ref{tbl:Gaussian_HD} and \ref{tbl:Gaussian_HD2}, the complexity of LP method can be $10$ times higher than our method on $m \times n$ Gaussian matrices, where $m$ is large.

\begin{table}[t!]
\centering
\caption{\label{tbl:Gaussian_HD} Lower bound on $k$ and execution time (Gaussian Matrix with $n=1024$)}
\subfloat[Lower bound on $k$]{
\begin{threeparttable}
\small{
    \begin{tabular*}{0.45\textwidth}{@{\extracolsep{\fill}}cccccc@{}}
        \hline
        $ \text{matrix}\; A$  & Pick-$1$ & $k(UB(\alpha_2)\tnote{b})$ & $k(\alpha_1)$ & LP\tnote{a} \\
        \hline
        $102 \times 1024$  & 2   & 3   & 2   & 2 \\
        $205 \times 1024$  & 4   & 4   & 4   & 4 \\
        $307 \times 1024$  & 5   & 6   & 5   & 5 \\
        $410 \times 1024$  & 7   & 8   & 7   & 7 \\
        $512 \times 1024$  & 9   & 10  & 9   & 9 \\
        $614 \times 1024$  & 12  & 13  & 12  & 12 \\
        $717 \times 1024$  & 16  & 17  & 15  & 16 \\
        $819 \times 1024$  & 21  & 23  & 20  & 21 \\
        $922 \times 1024$  & 32  & 36  & 30  & 32 \\
        \hline
    \end{tabular*}
}
\end{threeparttable}
}\\
\subfloat[Execution time (Unit: second)]{
\centering
\begin{threeparttable}
\small{
    \begin{tabular*}{0.45\textwidth}{@{\extracolsep{\fill}}ccccc@{}}
        \hline
        $ \text{matrix}\; A $  & Pick-$1$ & $k(UB(\alpha_2)\tnote{b})$ & $k(\alpha_1)$ &
        LP\tnote{a} \\
        \hline
        $102 \times 1024$    & 237   & 24 hours  & 237  & 200 \\
        $205 \times 1024$    & 452   & 24 hours  & 452   & 429  \\
        $307 \times 1024$    & 796   & 24 hours  & 796   & 723  \\
        $410 \times 1024$    & 1207  & 24 hours  & 1207  & 1073 \\
        $512 \times 1024$    & 1952  & 24 hours  & 1952  & 1600 \\
        $614 \times 1024$    & 2150  & 24 hours  & 2150  & 2217 \\
        $717 \times 1024$    & 1337  & 24 hours  & 1337  & 2992 \\
        $819 \times 1024$    & 838   & 24 hours  & 838   & 3904 \\
        $922 \times 1024$    & 386   & 24 hours  & 386   & 4730 \\
        \hline
    \end{tabular*}
}
\end{threeparttable}
}\\[-15pt]
\begin{threeparttable}
\begin{tabular*}{\textwidth}{c}
\end{tabular*}
\end{threeparttable}
\begin{tablenotes}[flushleft]
    \item [a] {\scriptsize Linear Programming \cite{juditsky2011verifiable}} \\
    \item [b] {\scriptsize Upper bound of $\alpha_2$ obtained from $1$-Step TSA after 24 hours' run}\\
\end{tablenotes}
\end{table}

\begin{table}[t!]
\centering
\caption{\label{tbl:Gaussian_HD2} Lower bound on $k$ and execution time (Gaussian Matrix)}
\subfloat[Lower bound on $k$]{
\begin{threeparttable}
{\small
    \begin{tabular*}{0.45\textwidth}{@{\extracolsep{\fill}}cccccc@{}}
        \hline
        $\text{matrix}\;A$  & Pick-$1$ & $k(\alpha_1)$ &
        LP\tnote{a}   \\
        \hline
        $512  \times 2048$   & 7   & 6    & 7 \\
        $2007 \times 2048$   & 102 & 90  & 102 \\
        $4014 \times 4096$   & 152 & 139 & N/A\tnote{b}\\
        $1024 \times 6144$   & 8   & 8    & 8   \\
        $6021 \times 6144$   & 190 & 174 & N/A \\
        $6134 \times 6144$   & 558 & 406 & N/A \\
        \hline
    \end{tabular*}
}
\end{threeparttable}
}\;\;
\subfloat[Execution time (Unit: second)]{
\begin{threeparttable}
{\small
    \begin{tabular*}{0.45\textwidth}{@{\extracolsep{\fill}}cccccc@{}}
        \hline
        $\text{matrix}\;A$  & Pick-$1$    & $k(\alpha_1)$
        & LP    \\
        \hline
        $512  \times 2048$   & 7.51e3 & 7.51e3  & 6.63e3 \\
        $2007 \times 2048$   & 6.71e2 & 6.71e2  & 7.19e4 \\
        $4014 \times 4096$   & 9.12e3 & 9.12e3  & 15 days\tnote{c} \\
        $1024 \times 6144$   & 2.18e5 & 2.18e5  & 1.61e5     \\
        $6021 \times 6144$   & 3.89e4 & 3.89e4  & 65.5 days\tnote{d} \\
        $6134 \times 6144$   & 1.37e4 & 1.37e4  & 41.7 days\tnote{e} \\
        \hline
    \end{tabular*}
}
\end{threeparttable}
}\\[-15pt]
\begin{threeparttable}
\begin{tabular*}{0.45\textwidth}{c}
\end{tabular*}
\end{threeparttable}
\begin{tablenotes}
    \item [a] {\scriptsize Linear Programming \cite{juditsky2011verifiable}}
    \item [b] {\scriptsize Not Available} \\
    \item [c] {\scriptsize Estimated time (15 hours for 4\% calculations) } \par
    \item [d] {\scriptsize Estimated time (15 hours for 1\% calculations) } \par
    \item [e] {\scriptsize Estimated time (10 hours for 1\% calculations) } \par
\end{tablenotes}
\end{table}

For extremely large sensing matrices, e.g, $4014 \times 4096$ and $6021 \times 6144$, the LP and SDP methods cannot provide any lower bound on $k$ due to unreasonable computational time. However, our pick-$l$ algorithm can still provide the lower bound on $k$ efficiently. Table \ref{tbl:Gaussian_HD2} shows the lower bound on $k$ and the execution time for these large dimensional matrices, where our verified recoverable sparsity $k$ can be as large as $558$ for a $6134 \times 6144$ sensing matrix. We obtained the estimated time for the LP method by running the Matlab code obtained from \url{http://www2.isye.gatech.edu/~nemirovs/}, which shows the percentage of the calculation on screen.

\subsection{Comparison between the optimized pick-$l$ algorithm and the basic pick-$l$ algorithm} 
We compared the basic pick-$l$ algorithm introduced in Section  \ref{Sec1} to the optimized pick-$l$ algorithm in Section \ref{Sec2} on Gaussian sensing matrices $28 \times 40$ and $40 \times 50$ for $l=3$ and $k=4,5,...,8$. Table \ref{tbl:Gaussian_LD_256} demonstrates that when $l=3$ and $k=4,5,...,8$, the optimized pick-$l$ algorithm provided tighter upper bounds on $\alpha_k$ than the basic pick-$l$ algorithm. This is because when $l$ is large and $k>l$, (15) includes more constraints, which leads to the reduced size of the feasible set, than the case when $k$ and $l$ are small. Hence, the optimal value of  (\ref{pickl_optimized_coefficients}), which is the result from the optimized pick-$l$, can be smaller than or equal to that of (\ref{optimized_coefficients_without_constraint}), which is the basic pick-$l$. Additionally, we provided the exact $\alpha_k$ values obtained from TSA in order to check how tight the bounds obtained from the basic pick-$l$ and the optimized pick-$l$ are. In terms of the execution time, the optimized pick-$l$ algorithm, which computes  (\ref{pickl_optimized_coefficients}), was around 1.7 and 4.4 times slower than the basic pick-$l$ on $28 \times 40$ and $40 \times 50$ Gaussian matrix respectively. \\
In summary, the optimized pick-$l$ algorithm provides better or at least equal upper bound on $\alpha_k$ to the basic pick-$l$ algorithm, with additional complexity. In spite of the increased complexity of the optimized pick-$l$ algorithm, it has an important theoretical merit, which is Lemma \ref{lemma3_various_coefficients}.

\begin{table}[t!]
\centering
\caption{\label{tbl:Gaussian_LD_256} $\alpha_k$ comparison and execution time (Gaussian Matrix)}
\subfloat[$\alpha_k$ comparison]{
\begin{threeparttable}
\setlength{\tabcolsep}{3.8pt}
{\scriptsize
    \begin{tabular}{ccccccc}
        \hline
        matrix $A$                  & Algo. & $\alpha_4$  & $\alpha_5$    & $\alpha_6$   & $\alpha_7$   & $\alpha_8$   \\
        \hline
        \multirow{3}{*}{28 $\times$ 40}    
        						   			 & Basic pick-$3$  		& 0.52  & 0.64  & 0.75  & 0.86   &  0.97       \\
        						   			 & Optimized pick-$3$  & 0.52  & 0.63  & 0.75  & 0.85   &  0.96    \\  
        						   			 & 3-Step TSA 				& 0.47  & 0.54  & 0.62  & 0.67   & 0.72-0.78 \\
         						    \hline       						   			 						   			 
        \multirow{3}{*}{40 $\times$ 50}    
        						   			 & Basic pick-$3$  		& 0.40  & 0.48  & 0.57  & 0.65   &  0.72       \\
        						   			 & Optimized pick-$3$  & 0.39  & 0.47  & 0.55  & 0.62   &  0.70    \\   
        						   			 & 3-Step TSA 				& 0.36  & 0.41  & 0.46  &  0.51   &  0.57-0.59 \\
        						    \hline         						    
    \end{tabular}
}
\end{threeparttable}
}\\
\subfloat[Execution time (Unit: second)]{
\centering
\setlength{\tabcolsep}{2.0pt}
\begin{threeparttable}
{\scriptsize
    \begin{tabular}{ccccccc}
        \hline
        matrix $A$                  & Algo. & $\alpha_4$  & $\alpha_5$    & $\alpha_6$   & $\alpha_7$   & $\alpha_8$   \\
        \hline
        \multirow{2}{*}{28 $\times$ 40}    
        						   			 & Basic pick-$3$  		& 249.28  & 249.28  & 249.28  & 249.28   & 249.28        \\
        						   			 & Optimized pick-$3$  & 420.97  & 410.43  & 422.14  & 422.41   & 460.52     \\     
         						    \hline       						   			 
        \multirow{2}{*}{40 $\times$ 50}    
        						   			 & Basic pick-$3$  		&  748.88 & 748.88  & 748.88  & 748.88   & 748.88        \\
        						   			 & Optimized pick-$3$  & 3.31e3  & 3.49e3  & 3.26e3  & 3.26e3   & 3.31e3     \\   
        						    \hline 
    \end{tabular}
}
\end{threeparttable}
}
\end{table}

\subsection{Complexity of tree search algorithm}
In this subsection, we carried out numerical experiments to demonstrate the computational complexity of TSA empirically on randomly chosen Gaussian sensing matrices. Fig. \ref{hist:SWA_TSA} (a) and (b) show the distribution of execution time and the distribution of number of nodes in height 5 attached to the tree structure in TSA respectively. For $m = 0.5n$, we generated 100 random realizations of Gaussian matrices and computed $\alpha_5$ using 3-Step TSA. The maximum number of leaf node whose cardinality is $k$ is $\binom{n}{k} = \binom{40}{5}=6.58008e5$. From Fig. \ref{hist:SWA_TSA} (b), we note that for 90 $\%$ of the cases, 3-Step TSA was terminated before 1.6 $\%$ of all the possible height-$5$ nodes were attached to the tree structure.

We provided the execution time of TSA for different-sized randomly chosen Gaussian matrices in Fig. \ref{fig:k_n_curve_n_increase}. We compared the execution time of TSA to ESM. Fig. \ref{fig:k_n_curve_n_increase} (a) shows that when $k=1$, 1-Step TSA provides almost similar performance to ESM. This is because 1-Step TSA calculates all the $\alpha_{1,\{i\}}$'s as a pre-computation, which is the same procedure as ESM. However, for $k>l$ as shown in Fig. \ref{fig:k_n_curve_n_increase} (b), (c), and (d), TSA can find $\alpha_k$ with reduced computation by using all the $\alpha_{l,L}$'s, while it is required to compute all the $\alpha_{k,K}$'s in ESM. In order to compute $\alpha_k$, we achieved a speedup of around 100 times via 2-Step TSA compared to ESM for $k=3,4$.

In addition, in Fig. \ref{fig:k_n_curve_k_increase},  we compared the execution time of TSA to ESM by varying $k$ with $n$ fixed on random Gaussian matrices. For the best execution time of TSA, we used different $l$ values for TSA. For $n=40$ and $n=50$, 3-Step TSA reduced the execution time to find $\alpha_5$ by around 100 times and 300 times respectively, compared with ESM .

\begin{figure}[t!]
    \centering
    \subfloat[]{ \includegraphics[scale=0.1]{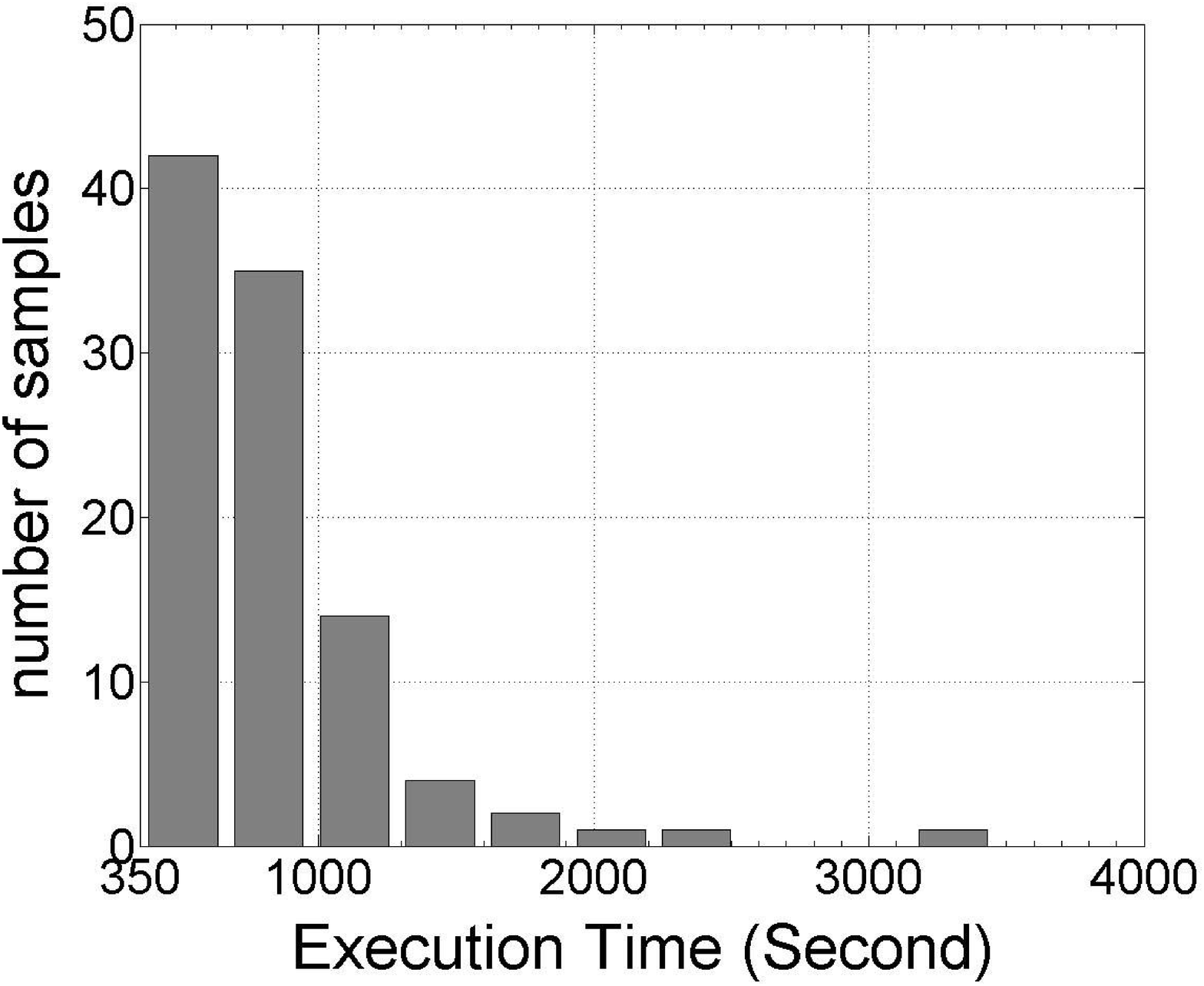} }
    \subfloat[]{ \includegraphics[scale=0.1]{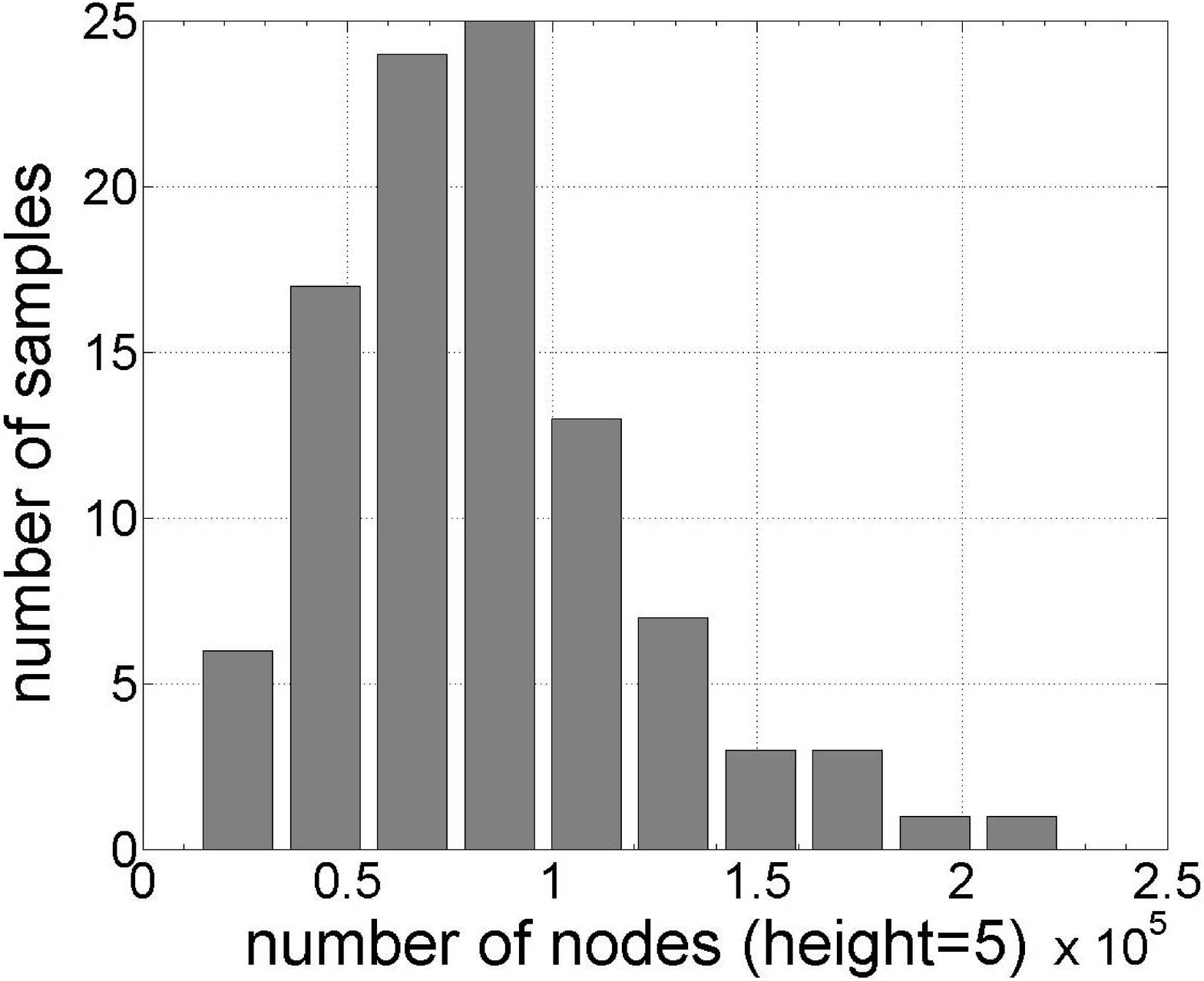} }
    \caption{\small{Histograms of the TSA (based on the pick-$3$ algorithm) to find $\alpha_5$ on 100 randomly chosen $20 \times 40$ Gaussian sensing matrices for each method. (a) Execution time. (b) Number of nodes in height 5. }}
    \label{hist:SWA_TSA}
\end{figure}
\begin{figure}[t!]
    \centering
    \subfloat[$k=1$]{ \includegraphics[scale=0.09]{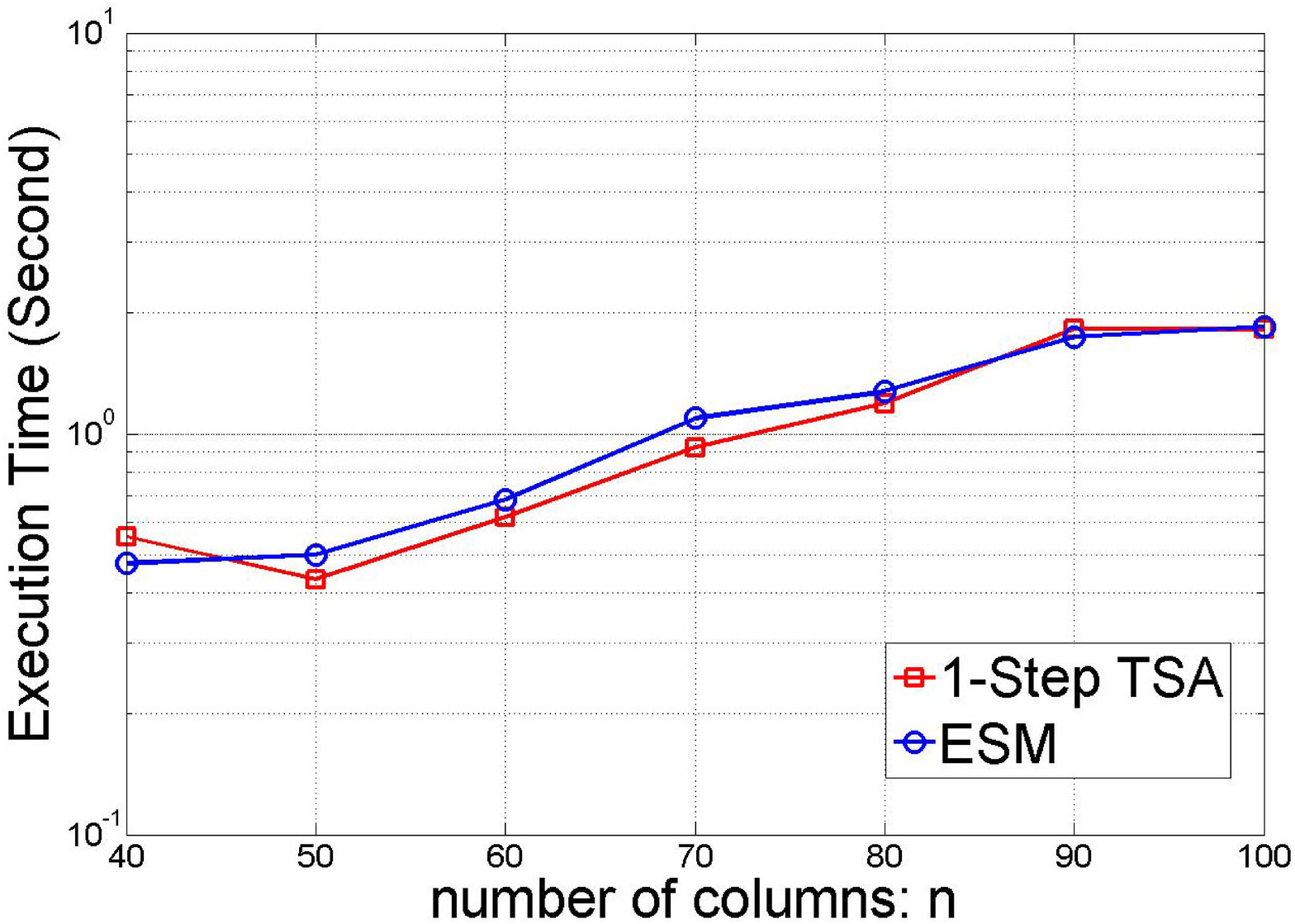} }
    \subfloat[$k=2$]{ \includegraphics[scale=0.09]{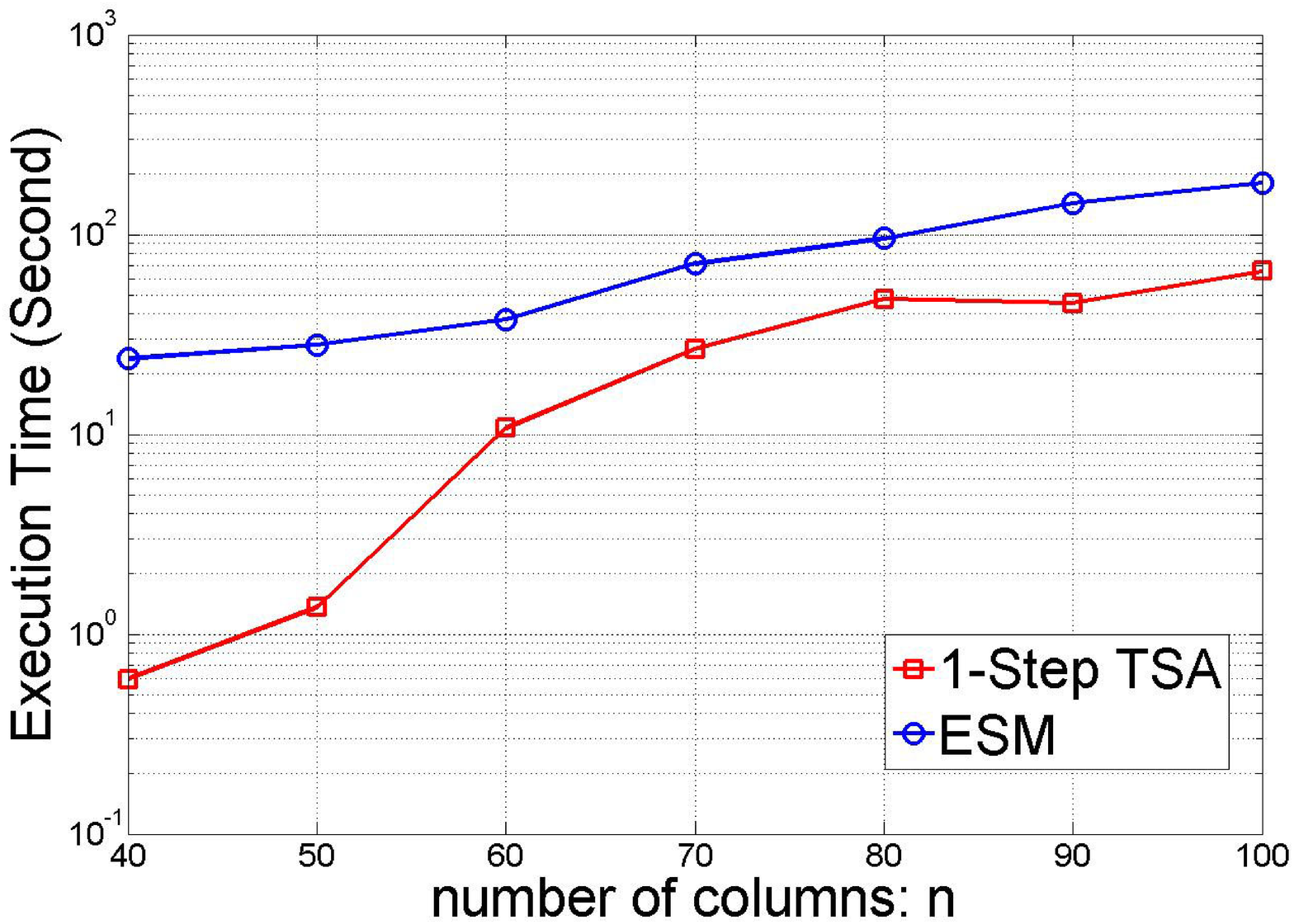} }\\
    \subfloat[$k=3$]{ \includegraphics[scale=0.09]{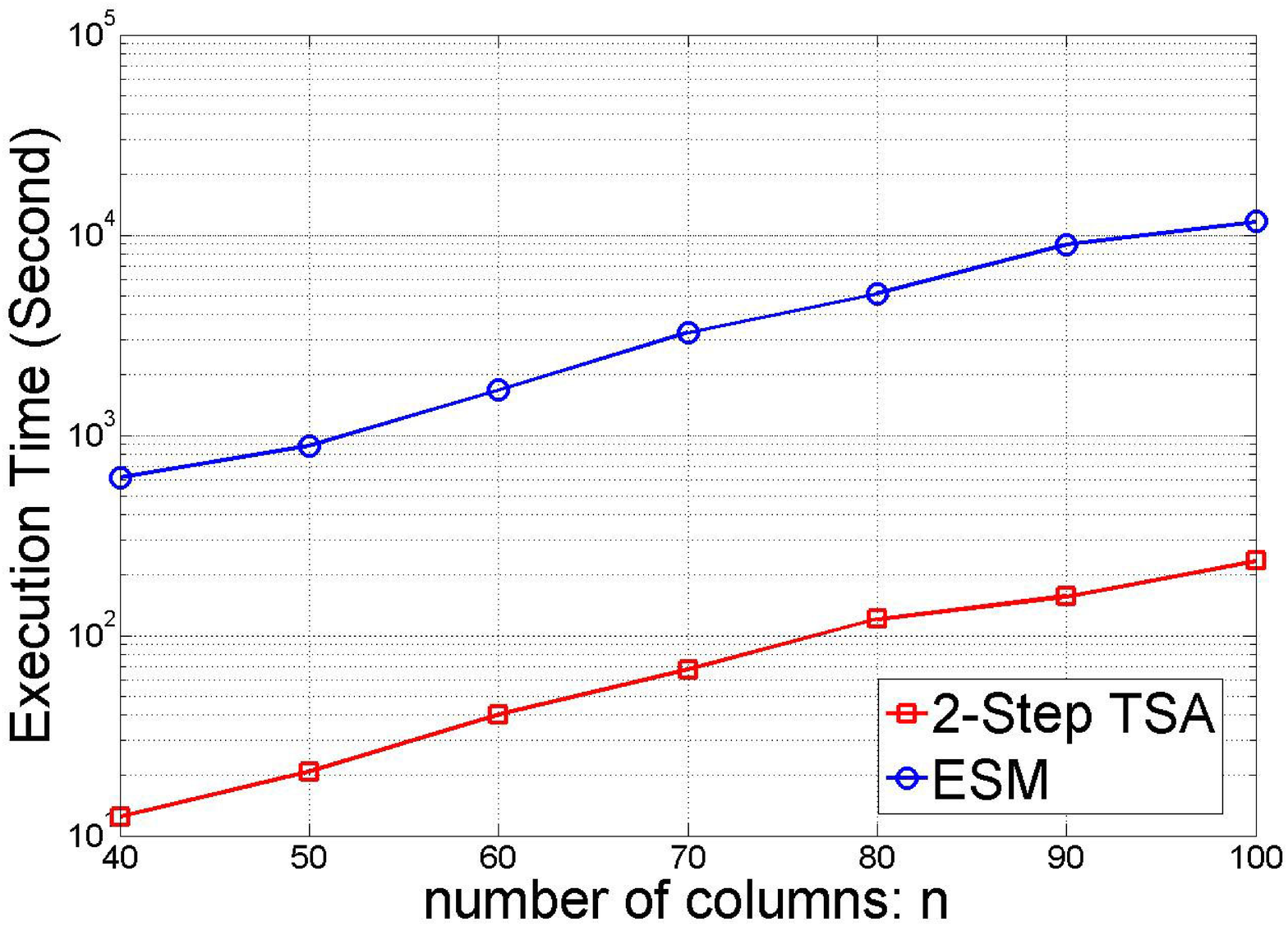} }
    \subfloat[$k=4$]{ \includegraphics[scale=0.09]{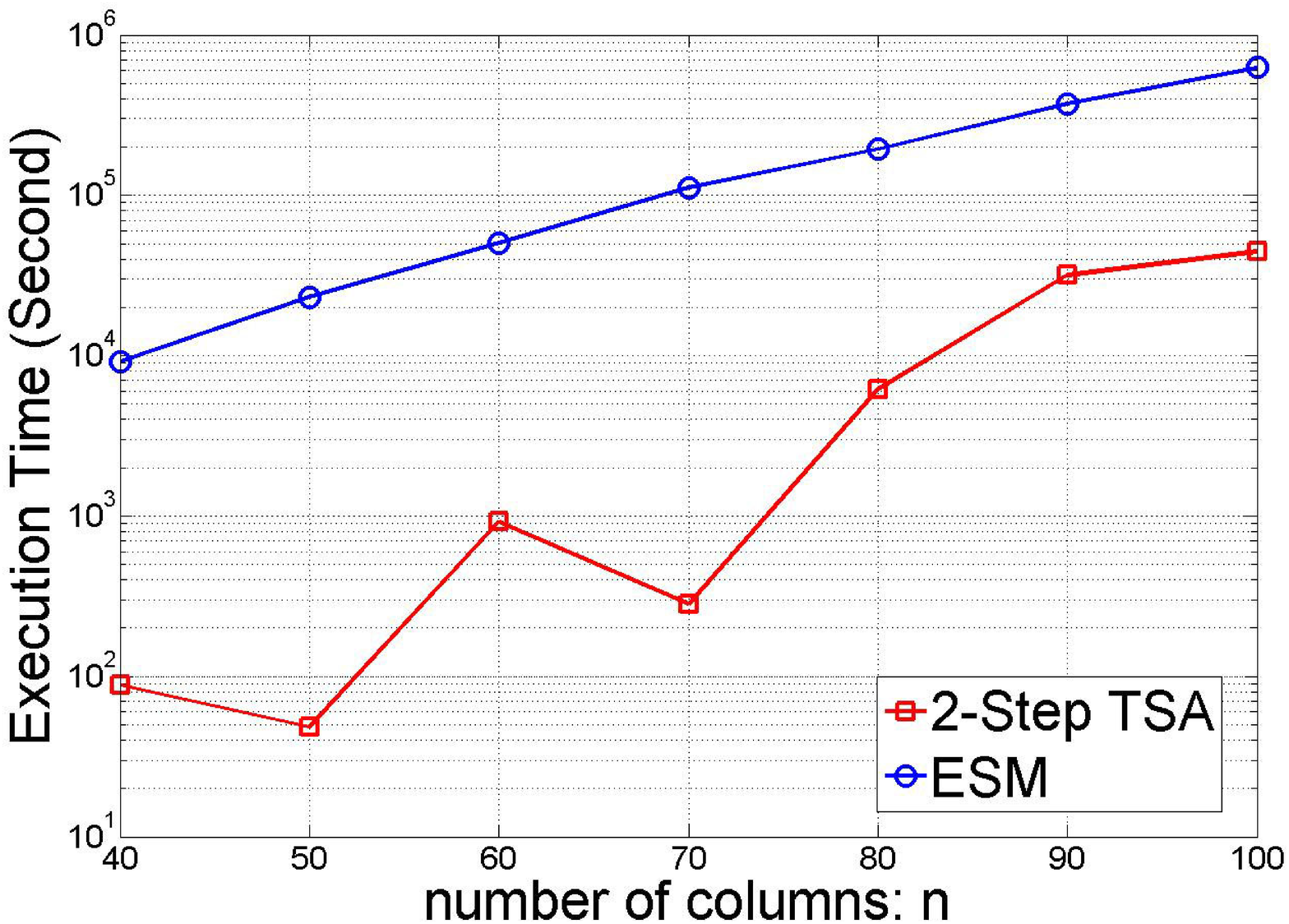} }
    \caption{\small{The execution time of TSA in log scale as a function of $n$ on randomly chosen $m \times n$ Gaussian matrices, where $m = n/2$.}}
    \label{fig:k_n_curve_n_increase}
\end{figure}
\begin{figure}[t!]
    \centering
    \subfloat[$n=40$]{ \includegraphics[scale=0.048]{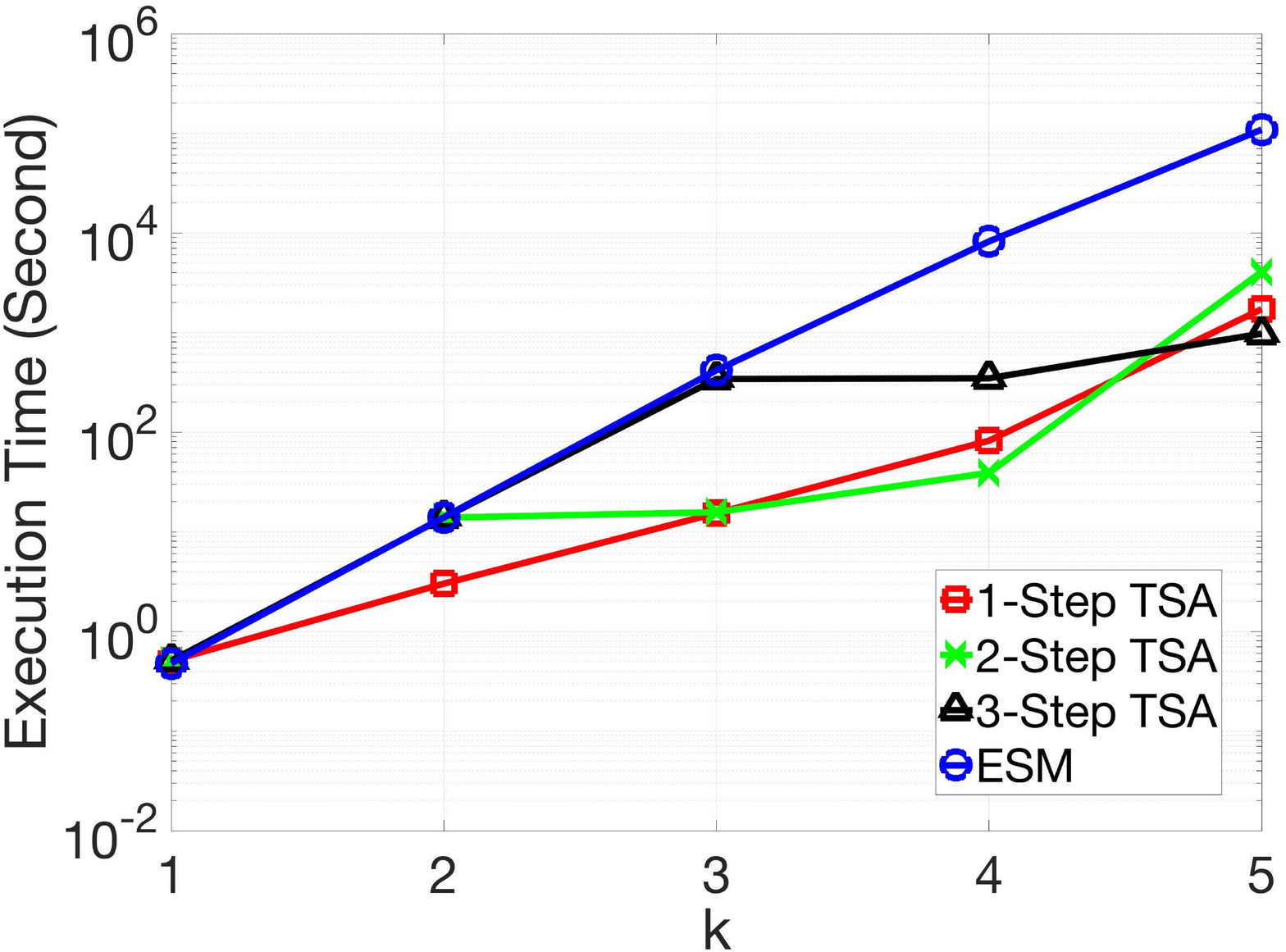} }
    \subfloat[$n=50$]{ \includegraphics[scale=0.048]{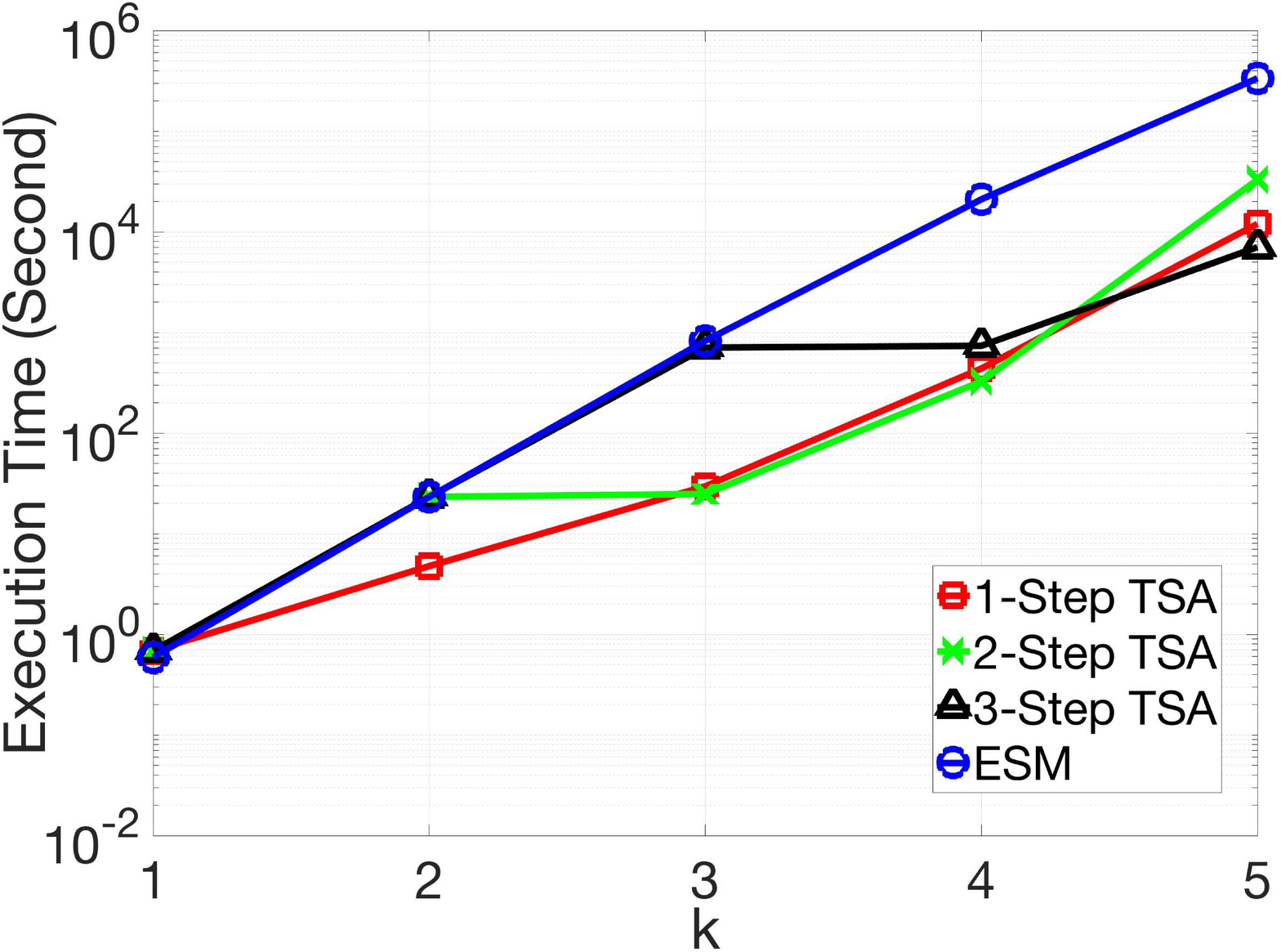} }
    \caption{\small{The execution time of TSA in log scale as a function of $k$ on randomly chosen $m \times n$ Gaussian matrices, where $m = n/2$.}}
    \label{fig:k_n_curve_k_increase}
\end{figure}

Finally, Fig. \ref{fig:TSA_GLB_GUB} gives illustrations of the values of the global lower and upper bounds, for $80 \times 100$ and $160 \times 200$ Gaussian sensing matrices, as the number of iterations in TSA increases. As we can see, the global upper and lower bounds get close very quickly. This implies that we can sometimes terminate TSA early and still obtain tight bounds on $\alpha_k$.

\begin{figure}[t!]
    \centering
    \subfloat[$80 \times 100$]{ \includegraphics[scale=0.22]{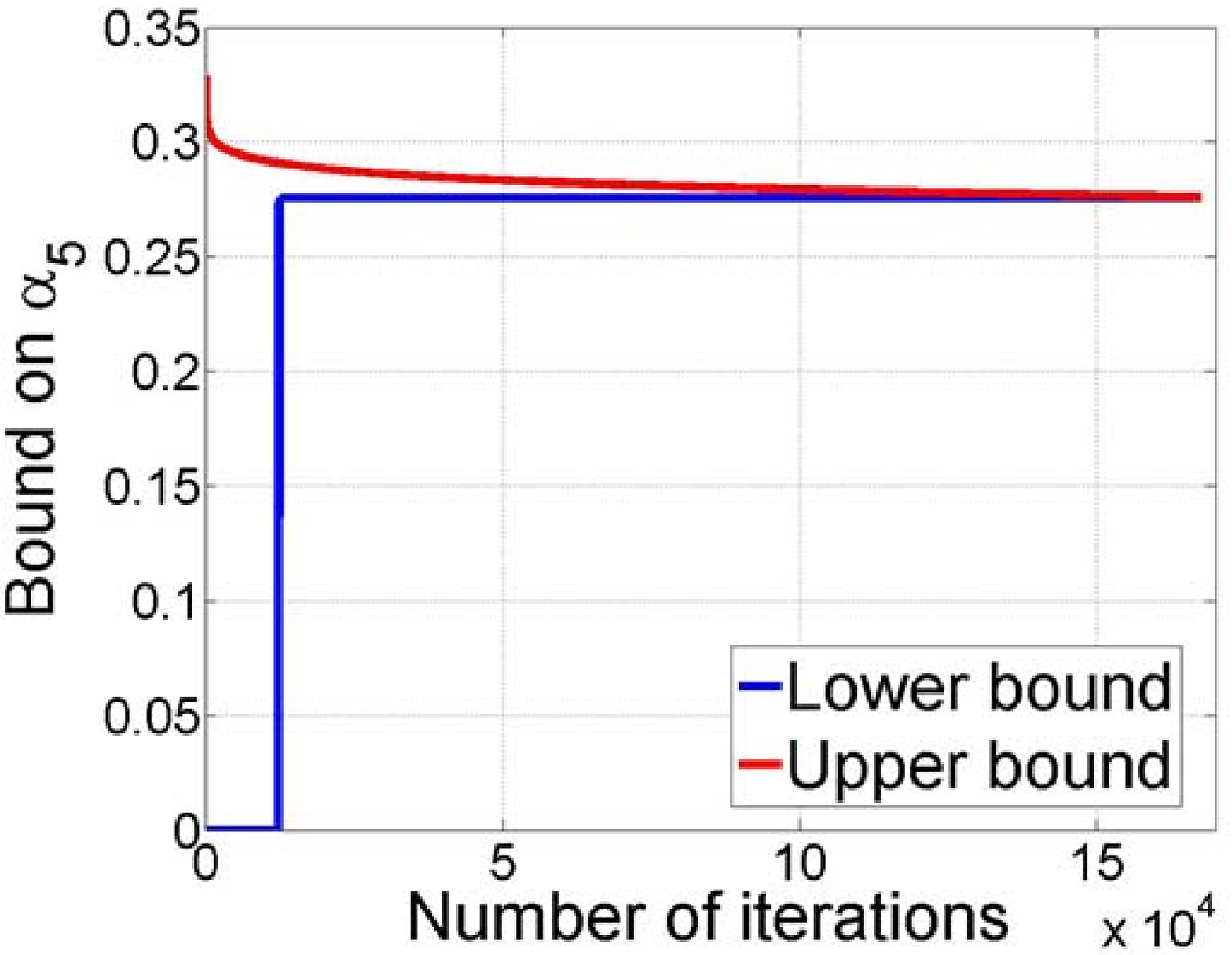} }
    \subfloat[$160 \times 200$]{ \includegraphics[scale=0.22]{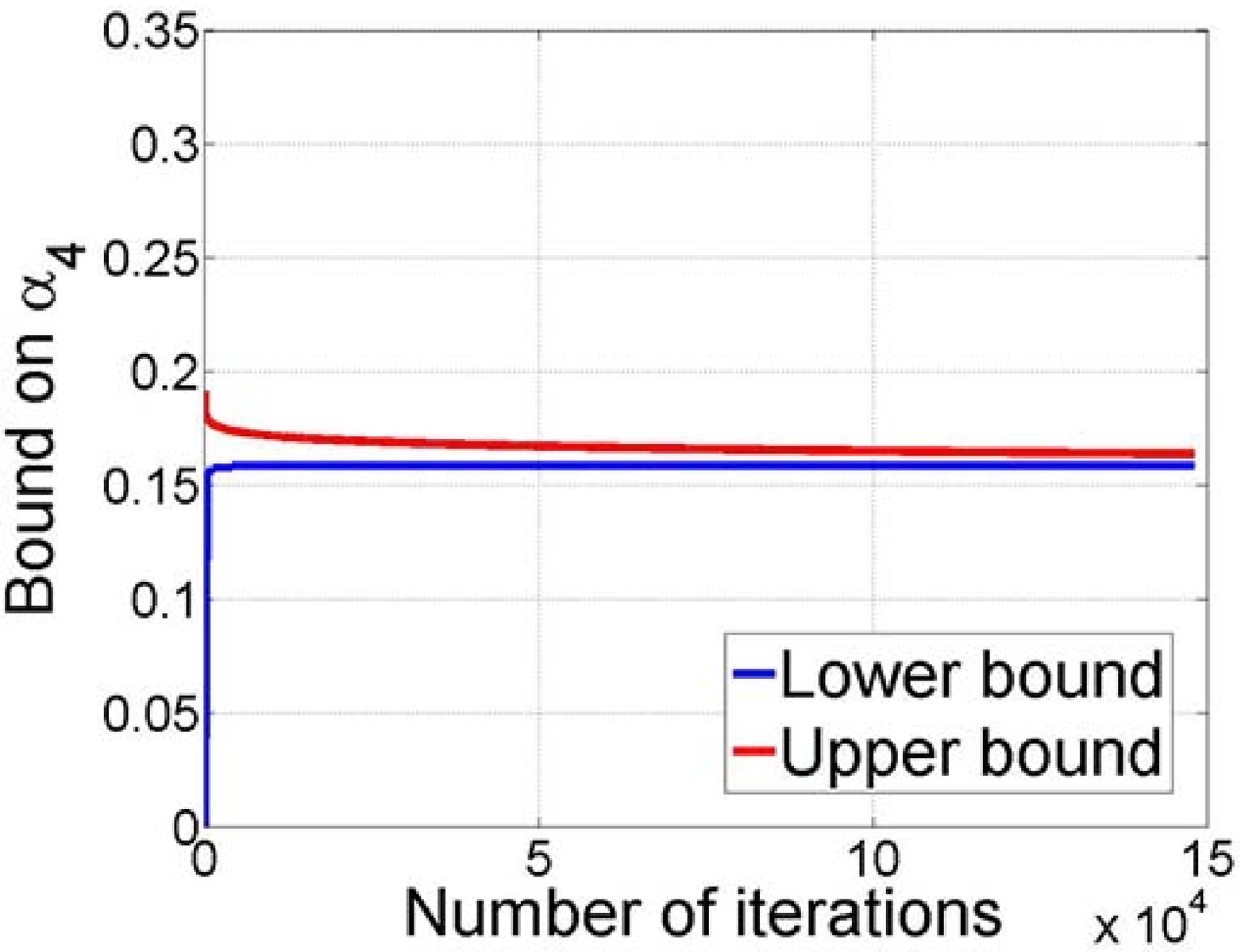} }
    \caption{\small{Global Lower Bound (GLB) and Global Upper Bound (GUB) in TSA on Gaussian sensing matrices. (a) For $(k, l) = (5, 3)$, we obtained (GLB,GUB)=(0.27,0.28) after 167501 iterations. (b) For $(k,l)=(4,2)$, we obtained (GLB, GUB)=(0.15,0.17) after 148101 iterations.}}
    \label{fig:TSA_GLB_GUB}
\end{figure}

\subsection{Application to network tomography problem}
We apply our new tools introduced in this paper to verify NSC for sensing matrices in network tomography problems \cite{xu2011Compressive,firooz2010network,coates2001network,tsang2003network,vardi1996network,castro2004network}. In an undirected graph model for the communication network, the communication delay over each link can be determined by sending packets through probing paths that are composed of connected links. The delay of each path is then measured by adding the delays over its links. Generally most links are uncongested, and only a few congested links have significant delays. It is, therefore, reasonable to think of finding the link delays as a sparse recovery problem. This sparse problem can be expressed in a system of linear equations $y=Ax$, where the vector $y \in \mathbb{R}^{m}$ is the delay of $m$ paths, the vector $x\in \mathbb{R}^{n}$ is the delay vector for the $n$ links, and $A$ is a sensing matrix. The element $A_{ij}$ of $A$ is $1$, if and only if path $y_i,\;i\in \{1,\;2,\;...,\;m\}$, goes through link $j$, $j \in \{1,\;2,\;...,\;n\}$; otherwise $A_{ij}$ equal to 0 (see Fig. \ref{network_model}). The indices of nonzero elements in the vector $x$ correspond to the congested links.
\begin{figure}[t!]
    \centering
    \includegraphics[scale=0.3]{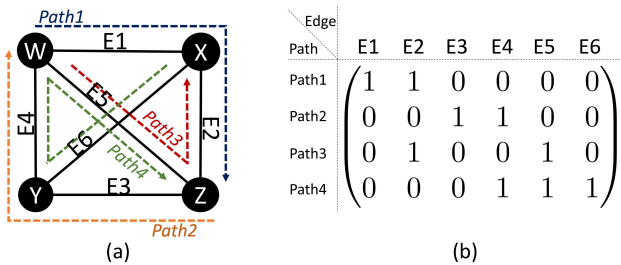}
    \caption{\small{(a) A simple example of a network tomography graph. $W$, $X$, $Y$, and $Z$ are nodes in the network, and \textit{Path$1$, $2$, $3$}, and \textit{$4$} are the probing paths through which the packets are sent. (b) The sensing matrix corresponding to the graph shown in (a). The rows and columns of the matrix represent probing paths and edges respectively.}}
    \label{network_model}
\end{figure}

In our numerical experiments to verify NSC in network tomography problems, the paths for sending data packets were generated by random walks of fixed length. Table \ref{network_example} summarizes the results of our experiments. We note that by using TSA, one can \emph{exactly} verify that a total of $k=2$ and $k=4$ congested link delays can be uniquely found by solving $\ell_1$ minimization problem (\ref{L1_min}) for the randomly generated network measurement matrices $33 \times 66$ (12-node complete graph) and $53 \times 105$ (15-node complete graph) respectively. For ESM, we estimated the execution time by multiplying the unit time to solve  (\ref{opt_prob2}) and the total number of cases in the exhaustive search. We obtained the unit time to solve (\ref{opt_prob2}) by calculating the arithmetic mean from 100 trials. For a $53 \times 105$ matrix, 3-Step TSA substantially reduced the execution time to find $\alpha_5$ around 137 times compared to ESM.

\begin{table}[t!]
\caption{\label{network_example} $\alpha_k$ and execution time in network tomography problems }
\centering
\subfloat[$\alpha_k$ values]{
\begin{threeparttable}
\setlength{\tabcolsep}{2pt}
{\small
    \begin{tabular*}{0.45\textwidth}{@{\extracolsep{\fill}}cccccccc@{}}
        \hline
        matrix $A$                  & Algo. & $\alpha_1$  & $\alpha_2$    & $\alpha_3$   & $\alpha_4$   & $\alpha_5$ & $k_{max}$  \\
        \hline
        \multirow{3}{*}{33 $\times$ 66}    & 1-Step TSA  & 0.28  & 0.41     & 0.50    & 0.57    & 0.62        & 2    \\
        						    & 2-Step TSA  & 0.28  & 0.41     & 0.50    & 0.57    & 0.62-0.64        & 2    \\
        						    & 3-Step TSA  & 0.28  & 0.41     & 0.50    & 0.57    & 0.62        & 2    \\
        						    \hline
        \multirow{3}{*}{53 $\times$ 105}   & 1-Step TSA  & 0.20  & 0.29     & 0.36    & 0.45    & 0.52-0.54        & 4    \\
        									  & 2-Step TSA  & 0.20  & 0.29     & 0.36    & 0.45    & 0.49-0.56        & 4    \\
        									  & 3-Step TSA  & 0.20  & 0.29     & 0.36    & 0.45    & 0.52        & 4    \\[-1pt]
        \hline
    \end{tabular*}
}
\begin{tablenotes}
    \item [a]{\scriptsize Random walk step: 20\quad}
\end{tablenotes}
\end{threeparttable}
}
\\
\subfloat[Execution time (Unit: second)]{
\begin{threeparttable}
\setlength{\tabcolsep}{2pt}
{\small
    \begin{tabular*}{0.45\textwidth}{@{\extracolsep{\fill}}ccccccc@{}}
        \hline
        matrix $A$         & Algo.          & $\alpha_1$ & $\alpha_2$    & $\alpha_3$   & $\alpha_4$   & $\alpha_5$  \\
        \hline
        \multirow{4}{*}{ \begin{tabular}[c]{@{}l@{}} 33 $\times$ 66    \end{tabular} }
             & 1-Step TSA   & 0.74  & 3.62  &  28.94  & 404.11   &  5.94e4 \\
             & 2-Step TSA   & 0.74  & 3.62  &  43.94  & 541.70   &  1 day\\ 
             & 3-Step TSA   & 0.74  & 3.62  &  1.69e3  & 1.73e3   &  3.70e4 \\
			& ESM  & 0.64 & 3.94  & 1.63e3  & 1.4e4\tnote{a}   & 1.8e5\tnote{a}           \\[-1pt]
        \hline
        \multirow{4}{*}{ \begin{tabular}[c]{@{}l@{}} 53 $\times$ 105  \end{tabular} }
             & 1-Step TSA   & 1.31  & 30.61 & 608.90 &  5.35e3 & 1 day        \\
             & 2-Step TSA   & 1.31  & 116.12 	&  143.99  & 1.05e3  & 1 day        \\
             & 3-Step TSA   & 1.31  & 116.12 	&  7.95e3 	 &  7.93e3	 &  1.38e4       \\
             & ESM          & 1.28 & 127.28 & 8.70.e3  & 9.6e4\tnote{a} & 1.9e6\tnote{a}        \\[-1pt]
        \hline
    \end{tabular*}
}
\begin{tablenotes}
    \item [a]{\scriptsize Exhaustive search method (Estimated execution time = average time to solve (\ref{opt_prob2}) (=0.02 second) for an index set $\times$ total number of index sets)}
\end{tablenotes}
\end{threeparttable}
}
\end{table}

We further carried out numerical experiments on even larger network model having 300 nodes and 400 edges. We created a random spanning tree for a network model by using random walk approach \cite{wilson1996generating}. At each probing path, we randomly chose a node among 300 nodes as a starting point of random-walk and walked 100 times along the network connection. We obtained a $320 \times 400$ matrix corresponding to the network model.  We calculated $\alpha_k$ values via $l$-Step TSA, where $l=1,2$. In terms of the execution time, in Table \ref{network_IEEE}, we compared TSA with ESM, where the unit time to solve (\ref{opt_prob2}) was obtained by calculating the arithmetic mean from 100 trials.  Especially, 1-Step TSA reduced the execution time to find $\alpha_4$ by around $28700$ times compared to ESM. 
\begin{table}[t!]
\caption{\label{network_IEEE}  $\alpha_k$ and execution time in a large network model having 300 nodes and 400
edges}
\centering
\subfloat[$\alpha_k$ values ]{
\begin{threeparttable}
\setlength{\tabcolsep}{2pt}
{\small
    \begin{tabular*}{0.45\textwidth}{@{\extracolsep{\fill}}ccccccc@{}}
        \hline
       Algo. & $\alpha_1$  & $\alpha_2$    & $\alpha_3$   & $\alpha_4$   & $\alpha_5$ &  $\alpha_6$\\
        \hline
        1-Step TSA  & 0.07  & 0.13     &  0.15   &  0.18   &  0.20    &    0.22-0.26\tnote{a} \\
        2-Step TSA  & 0.07  & 0.13     &  0.15   &  0.18   &  0.20-0.23\tnote{a}  &   0.22-0.28\tnote{a}  \\
        \hline
    \end{tabular*}
}
\end{threeparttable}
}
\\
\subfloat[Execution time (Unit: second)]{
\begin{threeparttable}
\setlength{\tabcolsep}{2pt}
    \begin{tabular*}{0.45\textwidth}{@{\extracolsep{\fill}}ccccccc@{}}
        \hline
        Algo.          & $\alpha_1$ & $\alpha_2$    & $\alpha_3$   & $\alpha_4$   & $\alpha_5$  & $\alpha_6$ \\
        \hline
             1-Step TSA   & 63.37  & 65.70     & 599.96   & 5.49e3  & 8.60e4  & 1 day\\
             2-Step TSA   & 63.37  &  3.46e4 & 3.54e4   &  4.03e4  & 1 day & 1 day\\ 
		    ESM & 73.22  & 3.20e4  & 1.59e6\tnote{b}   &  1.58e8\tnote{b}   &  1.25e10\tnote{b}  & 8.22e11\tnote{b}        \\[-1pt]
        \hline
    \end{tabular*}
\begin{tablenotes}
 	\item [a]{\scriptsize Lower bound - upper bound} \\
    \item [b]{\scriptsize Exhaustive search method (Estimated Operation time = average time to solve (\ref{opt_prob2}) (=0.15 second) for an index set $\times$ total number of index sets)}

\end{tablenotes}
\end{threeparttable}
}
\end{table}


\subsection{Discussion}
In this section, we discuss the strengths and weaknesses of our proposed algorithms, compared with earlier research \cite{d2011testing,juditsky2011verifiable}.
\begin{enumerate}
\item \textbf{Comparisons with LP and SDP.}  Our proposed pick-$1$-element algorithm can achieve similar performance as the LP \cite{juditsky2011verifiable} and SDP methods \cite{d2011testing}. However, our pick-$1$-element algorithm has the clear advantage of being more computationally efficient for large dimensional sensing matrices. Please see Table \ref{tbl:Gaussian_HD2}, where the LP and SDP methods cannot provide the performance bounds on recoverable sparsity $k$ due to high computational complexity. On the other hand, in Table \ref{tbl:Gaussian_HD2}, our pick-$1$-element algorithm can efficiently provide bounds on recoverable sparsity $k$. The LP method has high computational complexity because it has to deal with a large convex program of design dimension $mn$, which leads to prohibitive computational complexity when $m$ and $n$ are large \cite{juditsky2011verifiable}.

    In our pick-$1$-element algorithm, we proposed the novel idea of sorting $\alpha_{{1, \{i\}}}$'s (see Lemma \ref{lemma_pickl_upper_bound}), which leads to improved performance bounds on $\alpha_k$ and recoverable sparsity $k$. This sorting idea, combined with Lemma \ref{lemma_pickl_upper_bound}, provides us with larger recoverable sparsity bound $k$, than purely using $\alpha_1$ for bounding recoverable $k$ in \cite[Section 4.2.B]{juditsky2011verifiable}.

\item \textbf{Set-specific upper bounds.} Our proposed pick-$l$-element algorithm ($l\geq 2$) is novel, and can provide improved bounds on $\alpha_k$ and recoverable sparsity $k$, using polynomial computational complexity in $n$ when $l$ is fixed. This approach is not practical when $l$ is large. However, pick-2-element and pick-$3$-element algorithm can already provide improved performance bounds, compared with the previous research \cite{juditsky2011verifiable,d2011testing}.

    The fact that we can obtain upper bounds on $\alpha_{k}$, based on the results of pick-$l$-element ($l\geq 2$) algorithm, is new and non-trivial (see Lemma \ref{lemma_pickl_upper_bound}, Lemma \ref{kmax_bound} and Lemma \ref{pickl_optimized_coefficients_upper_bound}). For example, if we know $\alpha_5\leq 0.22$, we can use Lemma \ref{kmax_bound} to obtain that $\alpha_{11}\leq 0.22\times 11/5<0.5$.

    Our pick-$l$-element algorithm can provide set-specific upper bound for $\alpha_{k,K}$, laying the foundation for our branch-and-bound TSA.

\item \textbf{Computational complexity of TSA.} We proposed TSA to find precise values for $\alpha_{k}$ with significantly reduced average-case computational complexity than ESM.  The computational complexity of TSA is dependent on $n$, sparsity $k$, and a chosen constant $l$. When $k$, $n$ and $l$ are large enough, finding $\alpha_k$ via TSA is still computationally expensive. In the worst case, TSA has the same computational complexity as ESM.  However, our extensive simulations ranging from Fig. \ref{hist:SWA_TSA} to Fig. \ref{fig:k_n_curve_k_increase} and from Table \ref{network_example} to Table \ref{tbl:Fourier} show that on average, TSA can greatly reduce the computational complexity of finding $\alpha_k$ compared with ESM.

    Moreover, since TSA maintains an upper bound and a lower bound of $\alpha_k$ during its iterations, one can always early terminate TSA, and still get improved performance bounds on $\alpha_k$ than the LP and SDP methods. We can use TSA to find an exact value of $\alpha_l$, where $l < k$, and then use Lemma \ref{kmax_bound} to bound $\alpha_{k}$.

\item \textbf{Use of data structures.} We used Object-Oriented Programming (OOP) to implement TSA in Matlab \cite{matlabOOP}, because the OOP makes it easy to handle tree-type structures. In OOP, we defined a class and created objects from the class to store property of each node $J$, e.g., $B(J)$, in the tree. In order to make a connection between two tree nodes, we used doubly linked list data structure as a part of the object. However, in case readers would like to implement the algorithm using alternative data structures, we have provided implementation-agnostic pseudocode of our algorithm in Algorithm \ref{TSA_algo}.

\item \textbf{Difference from phase transition works.} There has been extensive research on the phase transitions of various sparse recovery algorithms such as Basis Pursuit (BP), Orthogonal Matching Pursuit (OMP), and Approximate Message Passing (AMP) \cite{eldar2012compressed}. However, our research is different from the research on phase transition in two aspects. Firstly, our work and the previous works \cite{juditsky2011verifiable,d2011testing} are focusing on worst-case performance guarantee (recovering all the possible $k$-sparse signals), while the research on phase transition is considering the average-case performance guarantee for a single $k$-sparse signal with fixed support and sign pattern. Secondly, the phase transition bounds are mostly for random matrices. Hence, for a given deterministic sensing matrix,  phase transition results cannot be used for that particular matrix.

\end{enumerate}

\section{Conclusion}
In this paper, we consider the problem of verifying the null space condition in compressed sensing. Calculating the proportional parameter $\alpha_k$ that characterizes the null space condition of a sensing matrix is a non-convex optimization problem, and also known to be NP-hard in \cite{Tillmann2014Computational}. In order to verify the null space condition, we proposed novel and simple enumeration-based algorithms, which are called the basic and optimized pick-$l$ algorithms, to obtain upper bounds of $\alpha_k$. With these algorithms, we further designed a new algorithm called the tree search algorithm to gain a global solution to the non-convex optimization problem of verifying the null space condition. Numerical experiments show that our algorithms outperform the previously proposed algorithms \cite{d2011testing,juditsky2011verifiable} in performance as well as speed.

\section*{Declarations}
\footnotesize
\section*{Availability of data and material}
All the codes used for the numerical experiments are available at the following link: \\
\url{https://sites.google.com/view/myungcho/software/nsc}.
\section*{Acknowledgements}
We thank Alexandre d'Aspremont from CNRS at Ecole Normale Superieure, Anatoli Juditsky from Laboratoire Jean Kuntzmann
at Universite Grenoble Alpes, and Arkadi Nemirovski from Georgia Institute of Technology for helpful discussions and providing codes for the simulations in \cite{d2011testing} and \cite{juditsky2011verifiable}.

\section*{Funding}
The work of Weiyu Xu is supported by Simons Foundation 318608, KAUST OCRF-2014-CRG-3, NSF DMS-1418737 and NIH 1R01EB020665-01.

\section*{Competing interests}
The authors declare that they have no competing interests.

\section*{Authors' contributions}
Myung Cho and Weiyu Xu designed the algorithms. Myung Cho implemented the algorithms. Kumar Vijay Mishra checked the implementation of the algorithms and helped to polish the manuscript. All authors read and approved the final manuscript.


\normalsize
\appendix
\section{Proof of Proposition \ref{thm:LPvsPick1}}
\begin{proof}
Let us denote the sum of $k$ maximal magnitudes of elements of $x \in \mathbb{R}^n$ as 
$$ ||x||_{k,1}= \underset{|K|\leq k}{\text{maximize}} \sum_{i \in K} |x_i|.$$
We use $i_1$, $i_2$, ..., and $i_k$ (or $j_1$, $j_2$, ..., and $j_k$ ) to denote $k$ distinct integers between 1 and $n$. For a matrix, say $A$, we use $A_{i,j}$ to represent its element in the $i$-th row and $j$-th column.
\setlength{\mathindent}{0pt}
\par\noindent
\small
\begin{align}
    \alpha_{k}^{LP}
    &= \underset{Y=[y_1,...,y_n]\in \mathbb{R}^{m\times n}}{ \text{minimize}}\;\bigg\{ \underset{1\leq j \leq n}{ \text{maximize}}\; ||(I-Y^T A)e_j||_{k,1}\; \bigg\} \nonumber\\
   & = \underset{Y=[y_1,...,y_n]\in \mathbb{R}^{m\times n}}{ \text{minimize}}\;\bigg\{ \underset{ i_1, i_2, ..., i_k,  j}{ \text{maximize}}\; \sum_{t=1}^{k}|(I-Y^T A)_{i_{t},j}|\; \bigg\} \nonumber\\
   &\leq  \underset{Y=[y_1,...,y_n]\in \mathbb{R}^{m\times n}}{ \text{minimize}}\;\bigg\{ \underset{\substack{ {i_1, i_2, ..., i_k,}\\{ j_1, j_2, ...,j_k}}}{ \text{maximize}}\; \sum_{t=1}^{k}|(I-Y^T A)_{i_{t},j_{t}}|\; \bigg\} \nonumber\\
   &= \underset{Y=[y_1,...,y_n]\in \mathbb{R}^{m\times n}}{ \text{minimize}}\;\bigg\{ \underset{ i_1, i_2, ..., i_k}{ \text{maximize}}\; \sum_{t=1}^{k} ||e_{i_t} - A^T y_{i_t}||_{\infty} \; \bigg\}  \nonumber\\
   &= \underset{ i_1, i_2, ..., i_k}{ \text{maximize}}\;  \bigg\{\sum_{t=1}^{k}   (\underset{y_{i_t}\in \mathbb{R}^{m\times 1}}{ \text{minimize}}\;   ||e_{i_t} - A^T y_{i_t}||_{\infty})\;\bigg\},
   \label{eq:boundproof}
\end{align}
\normalsize
where we can exchange the order of ``maximize'' and ``minimize'' in the last equality because $||e_{i_t} - A^T y_{i_t}||_{\infty}$ only depends on $y_{i_t}$.

Moreover, according to the equations for ``$\alpha^i$''  between (4.29) and (4.30) in \cite{juditsky2011verifiable} (taking $\beta$ there to be $\infty$),
\setlength{\mathindent}{15pt}
\par\noindent
\small
\begin{align*}
   &\underset{y_{i_t}\in \mathbb{R}^{m\times 1}}{ \text{minimize}}\;   ||e_{i_t} - A^T y_{i_t}||_{\infty}\;\\
   &= \underset{x}{ \text{maximize}} \big\{ e_{i_t}^Tx\; : \; Ax=0, \; ||x||_1 \leq 1 \big\} \\
   &=\alpha_{1,\{i_t\}}.
\end{align*}
\normalsize


Combining this with (\ref{eq:boundproof}), $\alpha_{k}^{LP}$ is no bigger than the upper bound calculated by Lemma \ref{lemma_pickl_upper_bound} (based on the pick-$1$-element algorithm). Namely,
\begin{align}
\alpha_{k}^{LP} \leq \alpha_k^{pick1}.
\end{align}
\qed
\end{proof}

\section{Sensing matrices with $n=40$}
\label{sec:app:lowSim}
Here, we provide the numerical results for small sensing matrices with $n=40$ to compare our methods to LP \cite{juditsky2011verifiable} and SDP \cite{d2011testing} methods.

\begin{table}[t]
\centering
\caption{\label{tbl:Gaussian} $\alpha_k$ comparison and execution time - Gaussian Matrix}
\subfloat[$\alpha_k$ comparison]{
\begin{threeparttable}
\setlength{\tabcolsep}{5.5pt}
\scriptsize{
    \begin{tabular}{|c|c|c|c|c|c|c|c|}
        \hline
        $\underset{(m \times n)}{A}$  & Algo.  & $\alpha_1$    & $\alpha_2$    & $\alpha_3$   & $\alpha_4$   & $\alpha_5$ & $k_{max}$\tnote{d} \\
        \hline
        \multirow{9}{*}{  $20\times40$  }
          & pick-$1$        & 0.28 & 0.55 & 0.81 & 1 	      & 1 		& 1/1.1 \\
          & pick-$2$       & 0.28 &  0.45 & 0.66 & 0.85 & 1 		& 2/1.9 \\
          & pick-$3$       & 0.28 & 0.45 &  0.57 & 0.76 & 0.92	& 2/1.9  \\
          & 1-Step TSA   &  0.28 & 0.45 & 0.57 & 0.67 & 0.75 	& 2/1.9  \\
          & 2-Step TSA   & 0.28 & 0.45 & 0.57 & 0.67 & 0.75 	& 2/1.9  \\
          & 3-Step TSA   & 0.28 & 0.45 & 0.57 & 0.67 & 0.75 	& 2/1.9  \\
          & LP\tnote{a}   & 0.28  & 0.50 & 0.67 & 0.84 & 0.98	& 2/1.6  \\
          & SDP\tnote{b} & 0.28 & 0.49 & 0.66 & 0.81 & 0.95	& 2/1.8  \\
          & ESM\tnote{c} & 0.28 & 0.45 & 0.57 & 0.67 & 0.75 & 2/1.9 \\
        \hline
        \multirow{9}{*}{  24$\times$40  }
          & pick-$1$      & 0.23 & 0.46 & 0.67 & 0.87 & 1 		& 2/2.0   \\
          & pick-$2$      & 0.23 & 0.37 & 0.53 & 0.69 & 0.85 	& 2/2.1   \\
          & pick-$3$      & 0.23 & 0.37 & 0.46 & 0.61 & 0.75 	& 3/2.8   \\
          & 1-Step TSA   & 0.23 & 0.37 & 0.46 & 0.57 & 0.65 	& 3/2.8  \\
          & 2-Step TSA   & 0.23 & 0.37 & 0.46 & 0.57 & 0.65 	& 3/2.8   \\
          & 3-Step TSA   & 0.23 & 0.37 & 0.46 & 0.57 & 0.65 	& 3/2.8   \\
          & LP            		& 0.23 & 0.41 & 0.56 & 0.71 & 0.84	& 2/2.0   \\
          & SDP           	& 0.23 & 0.41 & 0.55 & 0.70 & 0.82 	& 2/2.0	    \\
          & ESM           	& 0.23 & 0.37 & 0.46 & 0.57 & 0.65 	& 3/2.8   \\
        \hline
        \multirow{9}{*}{  28$\times$40  }
          & pick-$1$      & 0.18 & 0.36 & 0.53 & 0.70 & 0.86 &  2/2.0   \\
          & pick-$2$      & 0.18 & 0.31 & 0.46 & 0.59 & 0.72 &   3/3.0  \\
          & pick-$3$      & 0.18 & 0.31 & 0.41 & 0.54 & 0.66 &   3/3.0  \\
          & 1-Step TSA   & 0.18 & 0.31 & 0.41 & 0.49 & 0.57 &   4/3.5  \\
          & 2-Step TSA   & 0.18 & 0.31 & 0.41 & 0.49 & 0.57 &   4/3.5  \\
          & 3-Step TSA   & 0.18 & 0.31 & 0.41 & 0.49 & 0.57 &   4/3.5  \\
          & LP            		& 0.18 & 0.34 & 0.49 & 0.61 & 0.72 &   3/3.0  \\
          & SDP           	& 0.18 & 0.34 & 0.48 & 0.60 & 0.71 &   3/3.0  \\
          & ESM           	& 0.18 & 0.31 & 0.41 &  0.49 & 0.57 &   4/3.5  \\
        \hline
        \multirow{9}{*}{  32$\times$40  }
          & pick-$1$      & 0.14 & 0.29 & 0.42 & 0.55 & 0.67 &   3/3.0  \\
          & pick-$2$      & 0.14 & 0.24 & 0.37 & 0.47 & 0.58 &   4/3.8  \\
          & pick-$3$      & 0.14 & 0.24 & 0.33 & 0.44 & 0.53 &   4/4.2  \\
          & 1-Step TSA   & 0.14 & 0.24 & 0.33 & 0.40 & 0.47 &   5/4.9  \\
          & 2-Step TSA   & 0.14 & 0.24 & 0.33 & 0.40 & 0.47 &   5/4.9  \\
          & 3-Step TSA   & 0.14 & 0.24 & 0.33 & 0.40 & 0.47 &   5/4.9  \\
          & LP            		& 0.14 & 0.27 & 0.38 & 0.49 & 0.58 &   4/3.9  \\
          & SDP           	& 0.14 & 0.27 & 0.38 & 0.48 & 0.57 &   4/4.0  \\
          & ESM           	& 0.14 & 0.24 & 0.33 & 0.40 & 0.47 &   5/4.9   \\
        \hline
    \end{tabular}
}
\end{threeparttable}
}\\
\subfloat[Execution time]{
\begin{threeparttable}
\setlength{\tabcolsep}{5.3pt}
{\scriptsize
    \begin{tabular}{|c|c|c|c|c|c|c|}
        \multicolumn{7}{r} {\scriptsize (Unit: second) }\\
        \hline
        $\underset{(m \times n)}{A}$  & Algo.  & $\alpha_1$    & $\alpha_2$    & $\alpha_3$   & $\alpha_4$   & $\alpha_5$   \\
        \hline
        \multirow{9}{*}{  20$\times$40 }
          & pick-$1$      & 0.35  & 0.35    	& 0.35 		&  0.35 		& 0.35   \\
          & pick-$2$      & 0.35  & 10.96  	& 10.96  	& 10.96  	& 10.95   \\
          & pick-$3$      & 0.35  & 10.96  	& 313.65  	& 313.65   & 313.65    \\
          & 1-Step TSA   & 0.50  & 2.14    	& 11.78  	& 128.98  	& 1.62e3   \\
          & 2-Step TSA   & 0.50  & 13.20  	& 14.11  	& 58.93  	& 3.77e3   \\
          & 3-Step TSA   & 0.50  & 13.20  	& 320.20  & 346.43  & 695.53   \\          
          & LP            		& 0.55  & 0.55    	& 0.58  	& 0.55 		& 0.56   \\
          & SDP           	& 56.92 & 6.02e3  & 5.14e3  & 5.12e3  	& 5.61e3   \\
          & ESM           	& 0.35  & 10.96  	&  313.65  & 4.5e3 	& 6.0e4 \\
        \hline
        \multirow{9}{*}{  24$\times$40 }
          & pick-$1$      & 0.44  & 0.44  		& 0.44  	& 0.44  	& 0.44   \\
          & pick-$2$      & 0.44  & 13.00  	& 13.00  	& 13.00  	& 13.00   \\
          & pick-$3$      & 0.44  & 13.00  	& 311.27  	& 311.27  	& 311.27   \\
          & 1-Step TSA   & 0.50  & 2.05  		& 9.63  	& 77.45  	& 429.48   \\
          & 2-Step TSA   & 0.50  & 12.92  	& 13.60  	& 35.08  	& 634.62   \\
          & 3-Step TSA   & 0.50  & 12.92  	& 319.27  	& 378.10   & 481.29  \\
          & LP            		 & 0.84  & 0.94  	& 0.88  	& 0.83 		& 0.82  \\
          & SDP           	& 62.18  & 5.59e3  & 4.89e3 & 4.75e3  & 5.37e3   \\
          & ESM           	& 0.44  & 13.00  	& 311.27 	& 4.6e3 	& 6.4e4  \\
        \hline
        \multirow{9}{*}{  28$\times$40  }
          & pick-$1$      & 0.58  & 0.58  		& 0.58  	& 0.58  	& 0.58     \\
          & pick-$2$      & 0.58  & 14.67  	& 14.67  	& 14.67  	&  14.67  \\
          & pick-$3$      & 0.58  & 14.67  	& 326.80  & 326.80  & 326.80   \\
          & 1-Stpe TSA   & 0.52  & 1.41  		& 4.39  	& 32.43  	& 119.86   \\
          & 2-Stpe TSA   & 0.52  & 13.54  	& 13.82  	& 29.35  	& 126.62   \\
          & 3-Stpe TSA   & 0.52  & 13.54 	& 327.79  & 404.23  & 383.61   \\
          & LP            		& 1.12  	& 1.20  		& 1.12  		& 1.09  		&	0.68   \\
          & SDP           	& 71.27 & 5.55e3  & 4.90e3  & 4.98e3  & 4.72e3   \\
          & ESM           	& 0.58  & 14.67 	& 326.80  & 4.7e3 	& 6.9e4  \\
        \hline
        \multirow{9}{*}{ 32$\times$40 }
          & pick-$1$      	& 0.42  & 0.42  		& 0.42  		& 0.42  		& 0.42   \\
          & pick-$2$      	& 0.42  & 13.29  	& 13.29  	& 13.29  	& 13.29   \\
          & pick-$3$      	& 0.42  & 13.29  	& 331.80  & 331.80  	& 331.80   \\
          & 1-Step TSA   & 0.55  & 1.14  		& 2.89  	& 13.50  	& 40.67   \\
          & 2-Step TSA   & 0.55  & 14.22  	& 14.32  	& 18.13  	& 40.35   \\
          & 3-Step TSA   & 0.55  & 14.22  	& 340.87  & 336.29  & 355.06   \\
          & LP            		& 0.70  	&  0.71 		& 0.72 		& 0.70  		& 0.70   \\
          & SDP           	& 56.12 & 7.17e3  	& 5.43e3  & 5.07e3  	& 4.79e3   \\
          & ESM           	& 0.42  & 13.29  	& 331.80  & 4.9e3 	& 7.1e4   \\
        \hline
    \end{tabular}
}
\end{threeparttable}
}\\[-15pt]
\centering
\begin{threeparttable}
\begin{tabular*}{0.45\textwidth}{c}
\end{tabular*}
\end{threeparttable}
\begin{tablenotes}
    \item [a] {\scriptsize Linear Programming \cite{juditsky2011verifiable} }
    \item [b] {\scriptsize Semidefinite Programming \cite{d2011testing} }
    \item [c] {\scriptsize Exhaustive Search Method  }
    \item [d] {\scriptsize median / arithmetic mean }
\end{tablenotes}
\end{table}

\begin{table}[t]
\centering
\caption{\label{tbl:Fourier} $\alpha_k$ comparison and execution time - Partial Fourier Matrix}
\subfloat[$\alpha_k$ comparison]{
\begin{threeparttable}
{\scriptsize
    \begin{tabular}{|c|c|c|c|c|c|c|c|}
        \hline
        $\underset{(m \times n)}{A}$  & Algo.  & $\alpha_1$    & $\alpha_2$    & $\alpha_3$   & $\alpha_4$   & $\alpha_5$ & $k_{max}$\tnote{d}  \\
        \hline
        \multirow{9}{*}{  20$\times$40   }
          & pick-$1$      & 0.19 & 0.39 & 0.59 & 0.78 & 0.98 & 2/2.0   \\
          & pick-$2$      & 0.19 & 0.36 & 0.55 & 0.73 & 0.91 & 2/2.2   \\
          & pick-$3$      & 0.19 & 0.36 & 0.47 & 0.64 & 0.80 & 3/2.7   \\
          & 1-Step TSA    & 0.19 & 0.36 & 0.47 & 0.61 & 0.70 & 3/2.7   \\
          & 2-Step TSA    & 0.19 & 0.36 & 0.47 & 0.61 & 0.70 & 3/2.7   \\
          & 3-Step TSA    & 0.19 & 0.36 & 0.47 & 0.61 & 0.70 & 3/2.7   \\
          & LP\tnote{a}   & 0.19 & 0.39 & 0.59 & 0.78 & 0.98 & 2/2.0   \\
          & SDP\tnote{b}  & 0.19 & 0.39 & 0.59 & 0.78 & 0.98 & 2/2.0   \\
          & ESM\tnote{c}  & 0.19 & 0.36 & 0.47 & 0.61 & 0.70 & 3/2.7   \\
        \hline
        \multirow{9}{*}{  24$\times$40  }
          & pick-$1$      & 0.15 & 0.31 & 0.47 & 0.62 & 0.78 & 3/2.8   \\
          & pick-$2$      & 0.15 & 0.27 & 0.42 & 0.55 & 0.69 & 3/3.0   \\
          & pick-$3$      & 0.15 & 0.27 & 0.38 & 0.51 & 0.64 & 3/3.4   \\
          & 1-Step TSA    & 0.15 & 0.27 & 0.38 & 0.49 & 0.59 & 4/3.5   \\
          & 2-Step TSA    & 0.15 & 0.27 & 0.38 & 0.49 & 0.59 & 4/3.5   \\
          & 3-Step TSA    & 0.15 & 0.27 & 0.38 & 0.49 & 0.59 & 4/3.5   \\
          & LP            & 0.15 & 0.31 & 0.47 & 0.62 & 0.78 & 3/2.8   \\
          & SDP           & 0.15 & 0.31 & 0.47 & 0.62 & 0.78 & 3/2.8   \\
          & ESM           & 0.15 & 0.27 & 0.38 & 0.49 & 0.59 & 4/3.5   \\
        \hline
        \multirow{9}{*}{ 28$\times$40 }
          & pick-$1$      & 0.12 & 0.25 & 0.37 & 0.50 & 0.62 & 4/3.6  \\
          & pick-$2$      & 0.12 & 0.23 & 0.35 & 0.47 & 0.58 & 4/4.0  \\
          & pick-$3$      & 0.12 & 0.23 & 0.32 & 0.44 & 0.54 & 4/4.0  \\
          & 1-Step TSA    & 0.12 & 0.23 & 0.32 & 0.41 & 0.50 & 4/4.0  \\
          & 2-Step TSA    & 0.12 & 0.23 & 0.32 & 0.41 & 0.50 & 4/4.1  \\
          & 3-Step TSA    & 0.12 & 0.23 & 0.32 & 0.41 & 0.50 & 4/4.1  \\
          & LP            & 0.12 & 0.25 & 0.37 & 0.50 & 0.62 & 4/3.6  \\
          & SDP           & 0.12 & 0.25 & 0.37 & 0.50 & 0.62 & 4/3.6  \\
          & ESM           & 0.12 & 0.23 & 0.32 & 0.41 & 0.50 & 4/4.1  \\
        \hline
        \multirow{9}{*}{  32$\times$40  }
          & pick-$1$      & 0.09 & 0.19 & 0.29 & 0.38 & 0.48 & 5/4.7  \\
          & pick-$2$      & 0.09 & 0.17 & 0.27 & 0.36 & 0.44 & 5/4.7  \\
          & pick-$3$      & 0.09 & 0.17 & 0.25 & 0.35 & 0.43 & 5/4.7  \\
          & 1-Step TSA    & 0.09 & 0.17 & 0.25 & 0.33 & 0.39 & 5/4.7  \\
          & 2-Step TSA    & 0.09 & 0.17 & 0.25 & 0.33 & 0.39 & 5/4.7  \\
          & 3-Step TSA    & 0.09 & 0.17 & 0.25 & 0.33 & 0.39 & 5/4.7  \\
          & LP            & 0.09 & 0.19 & 0.29 & 0.38 & 0.48 & 5/4.7  \\
          & SDP           & 0.09 & 0.19 & 0.29 & 0.38 & 0.48 & 5/4.7  \\
          & ESM           & 0.09 & 0.17 & 0.25 & 0.37 & 0.39 & 5/4.7  \\
        \hline
    \end{tabular}
    }
\end{threeparttable}
}
\\
\subfloat[Execution time]{
\begin{threeparttable}
{\scriptsize
    \begin{tabular}{|c|c|c|c|c|c|c|}
        \multicolumn{7}{r} {\scriptsize (Unit: second) }\\
        \hline
        $\underset{(m \times n)}{A}$  & Algo.  & $\alpha_1$    & $\alpha_2$    & $\alpha_3$   & $\alpha_4$   & $\alpha_5$    \\
        \hline
        \multirow{9}{*}{ 20$\times$40  }
         & pick-$1$      & 0.31 & 0.31 & 0.31    & 0.31   & 0.31    \\
         & pick-$2$      & 0.31 & 10.85 & 10.85  & 10.85  & 10.85    \\
         & pick-$3$      & 0.31 & 10.85 & 260.41 & 260.41 & 260.41    \\
         & 1-Step TSA    & 0.47 & 9.72 & 70.57  & 329.28 & 3.60e3    \\
         & 2-Step TSA    & 0.47 & 11.97 & 18.54  & 45.18  & 3.36e3     \\
         & 3-Step TSA    & 0.47 & 11.97 & 291.29 & 297.45 & 633.12    \\
         & LP            & 0.49 & 0.77  & 0.53   & 0.59   & 0.51    \\
         & SDP           & 33.93 & 2.34e3 & 2.65e3 & 2.91e3 & 2.60e3    \\
         & ESM           & 0.31 & 10.85 & 260.41 & 4.1e3  & 6.0e4  \\
        \hline
        \multirow{9}{*}{  24$\times$40  }
         & pick-$1$      & 0.39 & 0.39  & 0.39   & 0.39  & 0.39   \\
         & pick-$2$      & 0.39 & 11.51 & 11.51  & 11.51 & 11.51    \\
         & pick-$3$      & 0.39 & 11.51 & 302.86 & 302.86 & 302.86    \\
         & 1-Step TSA    & 0.48 & 12.12 & 76.21 & 407.67 & 2.77e3    \\
         & 2-Step TSA    & 0.48 & 12.52 & 21.46  & 107.00 & 1.83e3    \\
         & 3-Step TSA    & 0.48 & 12.52 & 306.43 & 426.17 & 1.36e3    \\
         & LP            & 0.62 & 0.56  & 0.66   & 0.59   & 0.58    \\
         & SDP           & 41.13 & 2.39e3 & 2.66e3 & 2.63e3 & 2.56e3    \\
         & ESM           & 0.39 & 11.51 & 302.86 & 4.5e3 & 6.4e4  \\
        \hline
        \multirow{9}{*}{  28$\times$40 }
         & pick-$1$      & 0.43 & 0.43  & 0.43   & 0.43   & 0.43   \\
         & pick-$2$      & 0.43 & 13.29 & 13.29  & 13.29  & 13.29    \\
         & pick-$3$      & 0.43 & 13.29 & 341.05 & 341.05 & 341.05    \\
         & 1-Step TSA    & 0.50 & 8.70 & 31.53  & 272.68 & 731.90    \\
         & 2-Step TSA    & 0.50 & 12.99 & 16.85  & 47.45  & 544.79    \\
         & 3-Step TSA    & 0.50 & 12.99 & 317.40 & 410.47 & 553.67    \\
         & LP            & 0.65 & 0.67  & 0.71   & 0.67   & 0.75    \\
         & SDP           & 40.51 & 2.17e3 & 2.29e3 & 2.80e3 & 2.63e3    \\
         & ESM           & 0.43 & 13.29 & 341.05 & 4.7e3  & 6.5e4  \\
        \hline
        \multirow{9}{*}{  32$\times$40  }
         & pick-$1$      & 0.57 & 0.57  & 0.57    & 0.57   & 0.57   \\
         & pick-$2$      & 0.57 & 17.24 & 17.24   & 17.24  & 17.24    \\
         & pick-$3$      & 0.57 & 17.24 & 385.26  & 385.26 & 385.26    \\
         & 1-Step TSA    & 0.52 & 6.39 & 22.35   & 101.67 & 451.62    \\
         & 2-Step TSA    & 0.52 & 13.38 & 18.65   & 49.46  & 372.35    \\
         & 3-Step TSA    & 0.52 & 13.38 & 326.40  & 476.55 & 1.02e3    \\
         & LP            & 0.86 & 0.89  & 0.78 	 & 0.75   & 0.76    \\
         & SDP           & 46.51 & 2.41e3 & 2.62e3 & 2.53e3 & 2.75e3    \\
         & ESM           & 0.57 & 17.24 & 385.26  & 4.8e3  & 6.8e4     \\
        \hline
    \end{tabular}
    }
\end{threeparttable}
}\\[-15pt]
\centering
\begin{threeparttable}
\begin{tabular*}{0.45\textwidth}{c}
\end{tabular*}
\end{threeparttable}
\begin{tablenotes}
    \item [a] {\scriptsize Linear Programming \cite{juditsky2011verifiable} }
    \item [b] {\scriptsize Semidefinite Programming \cite{d2011testing} }
    \item [c] {\scriptsize Exhaustive Search Method  }
    \item [d] {\scriptsize median / arithmetic mean}
\end{tablenotes}
\end{table}

\end{document}